\documentclass[journal]{IEEEtran}

\ifCLASSINFOpdf
\else
\fi

\usepackage{amssymb}
\usepackage{paralist}
\usepackage{amsmath}
\usepackage{amsthm}
\usepackage{graphicx}
\usepackage{subfigure}
\usepackage{color}
\usepackage{bm}
\usepackage{algorithmic}
\usepackage{algorithm}
\usepackage{cite}
\newtheorem{Def}{Definition}
\newtheorem{Prop}{Proposition}

\newtheorem{Cor}{Corollary}

\newcommand{\eqdef}{\overset{def}{=}}
\allowdisplaybreaks

\begin{document}

\title{Dynamic Channel Access for Energy Efficient Data Gathering in Cognitive Radio Sensor Networks}

\author{Ju~Ren,~\IEEEmembership{Student Member,~IEEE,}
    Yaoxue~Zhang,
    Ning~Zhang,~\IEEEmembership{Student Member,~IEEE,}
    Deyu~Zhang,
    and~Xuemin (Sherman) Shen~\IEEEmembership{Fellow,~IEEE}
\IEEEcompsocitemizethanks{
\IEEEcompsocthanksitem Ju Ren, Yaoxue Zhang and Deyu Zhang are with the College of Information Science and Engineering, Central South University, Changsha, Hunan Province, China, 410083. Ju Ren and Deyu Zhang are also visiting scholars at the University of Waterloo now. E-mail: \{ren\_ju, zyx, zdy876\}@csu.edu.cn.
\IEEEcompsocthanksitem  Ning Zhang and Xuemin (Sherman) Shen are with the Department of Electrical and Computer Engineering, University of Waterloo, Waterloo, ON, Canada, N2L 3G1. E-mail: \{n35zhang, xshen\}@bbcr.uwaterloo.ca.}}

\maketitle

\begin{abstract}
Wireless sensor networks (WSNs) operating in the license-free spectrum suffer from uncontrolled interference as those spectrum bands become increasingly crowded. The emerging cognitive radio sensor networks (CRSNs) provide a promising solution to address this challenge by enabling sensor nodes to opportunistically access licensed channels. However, since sensor nodes have to consume considerable energy to support CR functionalities, such as channel sensing and switching, the opportunistic channel accessing should be carefully devised for improving the energy efficiency in CRSN. To this end, we investigate the dynamic channel accessing problem to improve the energy efficiency for a clustered CRSN. Under the primary users' protection requirement, we study the resource allocation issues to maximize the energy efficiency of utilizing a licensed channel for intra-cluster and inter-cluster data transmission, respectively. With the consideration of the energy consumption in channel sensing and switching, we further determine the condition when sensor nodes should sense and switch to a licensed channel for improving the energy efficiency, according to the packet loss rate of the license-free channel. In addition, two dynamic channel accessing schemes are proposed to identify the channel sensing and switching sequences for intra-cluster and inter-cluster data transmission, respectively. Extensive simulation results demonstrate that the proposed channel accessing schemes can significantly reduce the energy consumption in CRSNs.
\end{abstract}


\IEEEpeerreviewmaketitle

\section{Introduction}

\IEEEPARstart{W}{ireless} sensor network (WSN), as a promising event monitoring and data gathering technique, has been widely applied to various fields including environment monitoring, military surveillance and other industrial applications~\cite{deng2012energy}. A typical WSN consists of a large number of battery-powered sensor nodes to sense a specific area and periodically send the sensing results to the sink. Since sensor nodes are energy-constrained and generally deployed in unattended environment, energy efficiency becomes a critical issue in WSNs. Meanwhile, as the rapid growth of wireless services make the license-free spectrum increasingly crowded, WSNs operating over the license-free spectrum suffer from heavy interference caused by other networks sharing the same spectrum. The uncontrollable interference may cause a high packet loss rate and lead to excessive energy consumption for data retransmission, which significantly deteriorates the energy efficiency of the network.

Cognitive Radio (CR) technique has emerged as a promising solution to improve the spectrum utilization by enabling opportunistic access to the licensed spectrum bands~\cite{ning2014jsac}. This technology can also be applied to WSNs, which leads to Cognitive Radio Sensor Networks (CRSNs)~\cite{crsn2009}. Sensor nodes in CRSNs can sense the availability of licensed channels and adjust the operation parameters to access the idle ones, when the condition of the licensed-free channel degrades. However, since the energy consumed in supporting the CR functionalities, e.g., channel sensing and switching, is considerable for battery-powered sensor nodes~\cite{sensingcost_2010tvt,switchingcost_2013tvt}, the opportunistic channel access should be carefully studied to improve the energy efficiency in CRSNs. 

Existing works provide a comprehensive and in-depth investigation on optimizing the quality-of-service (QoS) performances for CRSNs, such as reducing the transmission delay~\cite{liang2011twc,adhoc2012delay,twc2014innetcomputations} or increasing the network capacity~\cite{routing2014tii,routing2014sensorjournal}. However, few of them have paid attention to improving the energy efficiency for CRSNs, with a delicate consideration of the energy consumption in channel sensing and switching. In order to enhance energy efficiency, the key issue is to determine when the energy consumption of transmitting a fixed amount of data can be reduced by sensing and accessing a licensed channel, compared with the energy consumption when only using the default license-free channel. It is very challenging since the decision depends on different factors, including the packet loss rate of the working channel, the probabilities for accessing licensed channels, as well as the protection for primary users (PUs). Moreover, due to the dynamic availability of licensed channels, when sensor nodes decide to sense and access a licensed channel, another challenge lies in identifying the best licensed channel to sense and access to optimize the energy efficiency for data transmission.

In this paper, we investigate the opportunistic channel accessing problem to improve energy efficiency in clustered CRSN. Sensor nodes form a number of clusters and periodically transmit their sensed data to the sink via hierarchical routing. They use license-free channel and are also able to access idle licensed channels when the packet loss rate increases. To protect the PUs sufficiently, the channel available duration (CAD) is limited for each licensed channel when it is detected as idle. Then, we analyze the expected energy consumption to determine if sensor nodes can reduce their energy consumption by accessing a licensed channel, considering the energy consumption in channel sensing and switching. Furthermore, to tackle the opportunistic availability of licensed channels, dynamic channel accessing with the resource allocation of an accessed channel is exploited for minimizing the energy consumption in both of the intra-cluster and inter-cluster data transmission. Specifically, the contributions of this work are three-fold.
\begin{enumerate}[(i)]
\item For both intra-cluster and inter-cluster data transmission, we determine the condition when sensor nodes should sense and switch to a licensed channel for potential energy consumption reduction.
\item We propose a dynamic channel accessing scheme to reduce the energy consumption for intra-cluster data transmission, which identifies the sensing and accessing sequence of the licensed channels within each cluster. 
\item Based on the analysis of intra-cluster data transmission, a joint power allocation and channel accessing scheme is developed for inter-cluster data transmission, which can dynamically adjust the transmission power of cluster heads and determine the channel sensing and accessing sequence to reduce energy consumption.
\end{enumerate}

The remainder of this paper is organized as follows. Section~\ref{sec2} overviews related works. The system model and problem statement are introduced in Section~\ref{sec3}. In Section~\ref{sec4}, we provide a detailed analysis of energy consumption for channel sensing decision and propose a dynamic channel sensing and accessing scheme for intra-cluster data transmission. Section~\ref{sec5} presents a joint power allocation and channel accessing scheme for inter-cluster data transmission. Simulation results are provided in Section~\ref{sec6} to evaluate the performance of the proposed schemes. Finally, Section~\ref{sec7} concludes the paper and outlines the future work.

\section{Related Works}
\label{sec2}
With ever-increasing wireless services and QoS requirements, traditional WSNs operating over the license-free spectrum, are facing unprecedented challenges to guarantee network performance. As an emerging solution for the spectrum scarcity of WSNs, CRSN has been well studied to improve the network performances, in terms of delay and throughput. 

Liang et al.~\cite{liang2011twc} analyze the delay performance to support real-time traffic in CRSNs. They derive the average packet transmission delay for two types of channel switching mechanisms, namely periodic switching and triggered switching, under two kinds of real-time traffic, including periodic data traffic and Poisson traffic, respectively. Bicen et al.~\cite{adhoc2012delay} provide several principles for delay-sensitive multimedia communication in CRSNs through extensive simulations. A greedy networking algorithm is proposed in~\cite{twc2014innetcomputations} to enhance the end-to-end delay and network throughput for CRSNs, by leveraging distributed source coding and broadcasting. Since the QoS performances of sensor networks can be significantly impacted by routing schemes, research efforts are also devoted in developing dynamic routing for CRSNs~\cite{routing2014tii,routing2014sensorjournal}. Quang and Kim~\cite{routing2014tii} propose a throughput-aware routing algorithm to improve network throughput and decrease end-to-end delay for a large-scale clustered CRSN based on ISA100.11a. In addition, opportunistic medium access (MAC) protocol design and performance analysis of existing MAC protocols for CRSNs are studied in~\cite{csma2014adhoc,mac2015tvt}.

Most of the existing works can effectively improve the network performances for various WSNs applications, and also provide a foundation for spectrum management and resource allocation in CRSNs. However, as a senor network composed of resource-limited and energy-constrained sensor nodes, CRSN is still facing an inherent challenge on energy efficiency, which attracts increasing attention to study the energy efficiency enhancement.

Han et al.~\cite{tvt2011ee} develop a channel management scheme for CRSNs, which can adaptively select the operation mode of the network in terms of channel sensing, channel switching, and data transmission/reception, for energy efficiency improvement according to the outcome of channel sensing. The optimal packet size is studied in~\cite{twc2012packetsize} to maximize energy efficiency while maintaining acceptable interference level for PUs and achieving reliable event detection in CRSNs. The transmission power of sensor nodes can also be adjusted for improving the energy efficiency of data transmission. In~\cite{powerallocation2012icc}, Chai et al. propose a power allocation algorithm for sensor nodes to achieve satisfactory performance in terms of energy efficiency, convergence speed and fairness in CRSNs. Meanwhile, since spectrum sensing accounts for a certain portion of energy consumption for CRSNs, energy efficient spectrum sensing schemes are also studied in CRSNs to improve the spectrum detection performance~\cite{sensorjournal2013sensingee, sensorjournal2011sensingee}. Furthermore, motivated by the superior energy efficiency of clustered WSNs, spectrum-aware clustering strategies are investigated in~\cite{tvt2014clustering, infocom2013clustering} to enhance energy efficiency and spectrum utilization for CRSNs. 

However, a comprehensive study on energy efficient data gathering is very important for CRSNs, which should jointly consider the energy consumption in channel sensing and switching, channel detection probability and PU protection to determine channel sensing and switching decision. 

\section{System Model}
\label{sec3}

\subsection{Network Model}
Consider a cognitive radio sensor network, where a set of cognitive sensor nodes $\mathcal{N} = \{s_1, ..., s_n\}$ are distributed to monitor the area of interest, as shown in Fig.~\ref{fig.network}. According to the application requirements, sensor nodes sense the environment with different sampling rates and periodically report their sensed data to the sink node~\cite{ju2014jcsu}. We divide the operation process of the network into a large number of data transmission periods and let $T$ denote the duration of a transmission period. Motivated by the benefits of hierarchical data gathering, sensor nodes form a number of clusters, denoted by $\mathcal{L} = \{L_1, ..., L_m\}$, to transmit the sensed data to the sink~\cite{ju2014ksii}. Denote the cluster head (CH) of $L_i$ as $H_i$, and the set of cluster members (CMs) in $L_i$ as $\mathcal{N}_i$.
\begin{figure}
\centering
\includegraphics[width=0.45\textwidth]{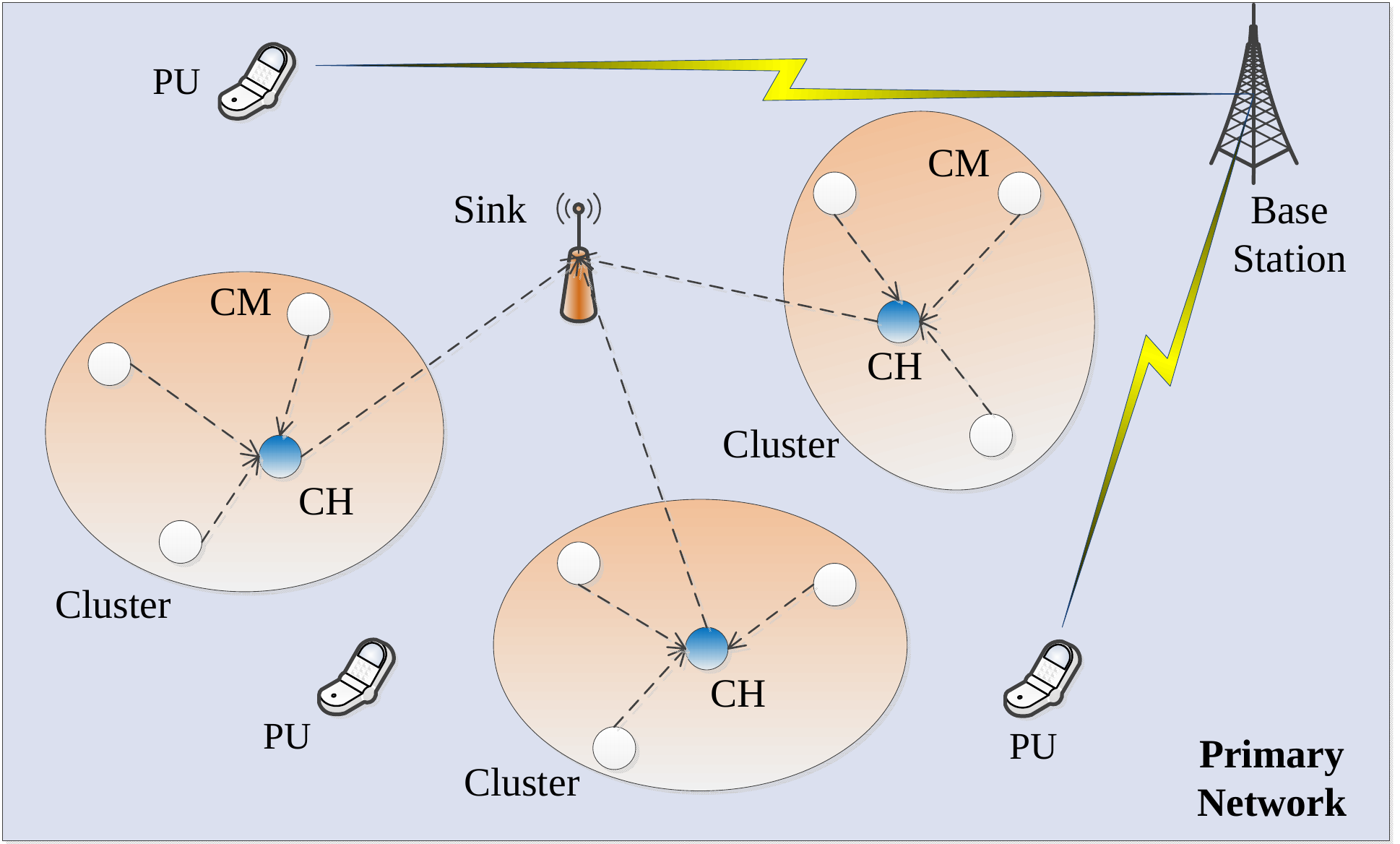}
\caption{The architecture of CRSN.}
\label{fig.network}
\end{figure}

The data transmission during each transmission period consists of two phases: intra-cluster data transmission and inter-cluster data transmission. In the intra-cluster data transmission, CMs directly transmit their sensed data to the cluster heads in a Time Division Multiple Access (TDMA) manner. During the inter-cluster data transmission, CHs aggregate the sensed data and directly send the aggregated intra-cluster data to the sink. The inter-cluster data transmission is also based on a TDMA manner, coordinated by the sink. The sensor network operates on a license-free channel $C_0$ for data transmission, which usually suffers from uncontrolled interference causing a significant packet loss rate. Enabled by the cognitive radio technique, sensor nodes can sense the licensed channels and access the vacant ones, when the packet loss rate of $C_0$ is fairly high. Sensor nodes have only one radio for data communication, which means sensor nodes can only access one channel at a time. %When a channel is used for the intra-cluster data transmission in neighboring clusters, we do not consider the co-channel interference between the clusters, which has been addressed in existing works~\cite{shen2011delay}.

\subsection{Cognitive Radio Model}
Suppose that there are $k$ different licensed data channels $\mathcal{C} = \{C_1, ..., C_k\}$ with different bandwidths $\{B_1, ..., B_k\}$ in the primary network. The PU's behavior is assumed to be stationary and ergodic over the $k$ channels. The cognitive sensor nodes in the primary network are secondary users (SUs) that can opportunistically access the idle channels. A fixed common control channel is considered to be available to exchange the control information among the sensor nodes and the sink. We model the PU traffic as a stationary exponential ON/OFF random process~\cite{ning2014jsac}. The ON state indicates that channel is occupied by PUs and the OFF state implies that the channel is idle. Let $V_x$ and $L_x$ be the exponential random variables, describing the idle and occupancy durations of $C_x$ with means $v_x$ and $l_x$, respectively. Thus, for each channel $C_x$, the probability of channel being idle $p_{off}^{x}$ and the probability of channel occupancy $p_{on}^{x}$ are
\begin{align}
\label{energymodel}
\begin{cases}
p_{off}^{x} =  v_x / (v_x + l_x), & \mathcal{H}_{0,x}\\ 
p_{on}^{x} =  l_x / (v_x + l_x), & \mathcal{H}_{1,x}
\end{cases},
\end{align}
where $\mathcal{H}_{0,x}$ and $\mathcal{H}_{1,x}$ represent the hypothesis that $C_x$ is idle and occupied, respectively. Sensor nodes are assumed to sense channel by the energy detection-based spectrum sensing approach~\cite{zhang2013cooperative}. When $s_j$ adopts energy detector to sense $C_x$, the detection probability $p_{d,x,j}$ (i.e., the probability of an occupied channel being determined to be occupied correctly) and the false alarm probability $p_{f,x,j}$  (i.e., the probability of an idle channel being determined as occupied) are defined as
% \begin{align}
% \label{energymodel}
% \begin{cases}
% p_{d,x,j} = Pr(D_x \ge \delta_x | \mathcal{H}_{1,x})\\ 
% p_{f,x,j} = Pr(D_x \ge \delta_x | \mathcal{H}_{0,x})
% \end{cases},
% \end{align}
 \begin{align*}
 \begin{split}
    p_{d,x,j} = Pr(D_x \ge \delta_x | \mathcal{H}_{1,x}),~p_{f,x,j} = Pr(D_x \ge \delta_x | \mathcal{H}_{0,x})
 \end{split}
 \end{align*}
where $\delta_x$ is the detection threshold and $D_x$ is the test statistic for $C_x$. %Particularly, $D_x =\frac{1}{M}\sum^{M}_{n=1} |y_x(n)|^2$, where $M$ is the number of samples in an observation period and $y_x(n)$ is the $n$-th sample of the received signal for $C_x$. For each sensor node, we use $\varphi$ to denote the spectrum sensing time, and assume that the received signal of the PUs is sampled at a sampling frequency $f_s$. Then, we have $M = \varphi f_s$. 
And the misdetection probability can be calculated as $p_{m,x,j} = Pr(D_x < \delta_x | \mathcal{H}_{1,x}) = 1-p_{d,x,j}$.

According to \cite{liang2008twc}, the false alarm probability of $s_j$ for $C_x$ can be given by $p_{f,x,j}=\textit{Q}\left ( \left (\frac{\delta_x}{\sigma_x^2}-1 \right )\sqrt{\varphi f_s} \right )$,
 % \begin{eqnarray*}
 %   p_{f,x,j}=\textit{Q}\left ( \left (\frac{\delta_x}{\sigma_x^2}-1 \right )\sqrt{\varphi f_s} \right )
 %      \label{eq:iss1}
 % \end{eqnarray*}
where $\textit{Q}(\cdot)$ is the complementary distribution function of the standard Gaussian. The detection probability of $s_j$ for $C_x$ is given by $p_{d,x,j} =\textit{Q}\left ( \left (\frac{\delta_x}{\sigma_x^2}-\overline{\gamma}_{x,j}-1 \right )\sqrt{\frac{\varphi f_s}{2 \overline{\gamma}_{x,j}+1}}\right )$,
% \begin{eqnarray*}
%   \begin{split}
%     p_{d,x,j} =\textit{Q}\left ( \left (\frac{\delta_x}{\sigma_x^2}-\overline{\gamma}_{x,j}-1 \right )\sqrt{\frac{\varphi f_s}{2 \overline{\gamma}_{x,j}+1}}\right )
%       \label{eq:iss2}
%     \end{split}
%  \end{eqnarray*}
where $\overline{\gamma}_{x,j}$ is the average received signal-to-noise ratio (SNR) over channel $C_x$ at $s_j$. %Particularly, $\overline{\gamma}_{x,j}=\frac{P_{PU}h_{x,j}^2}{\sigma_x^2}$, where $P_{PU}$ is the transmission power of the PU, $h_{x,j}$ is the average channel gain from channel $x$ to sensor $j$, and $\sigma_x^2$ is the variance of the Gaussian noise.

To enhance the accuracy of sensing results, sensor nodes collaboratively perform channel sensing. Specifically, sensor nodes in the same cluster send the individual sensing results to the cluster head to make a combined decision. The decision rules at the cluster head can include AND rule, OR rule, etc. When OR rule is adopted, PUs are considered to be present if at least one sensor claims the presence of PUs. Then, if we use a number of sensor nodes, e.g., a set of sensor nodes $\bm{y}$, to cooperatively sense a channel, the cooperative detection probability $F^x_d$ and the cooperative false alarm probability $F^x_f$ for channel $C_x$ are
\begin{align}
\begin{cases}
    F^x_d &= 1-\prod_{N_j \in \bm{y}}(1-p_{d,x,j})\\
    F^x_f &= 1-\prod_{N_j \in \bm{y}}(1-p_{f,x,j})
\end{cases}
\label{eq:css3}
\end{align}

The cooperative misdetection probability $F^x_m$ is defined as the probability that the presence of the PU is not detected, i.e., $F^x_m=1-F^x_d$. In order to guarantee the accuracy of spectrum sensing, channel sensing should satisfy a requirement that the probability of interfering with PUs should be below a predefined threshold $F_I$. In other words, there is a constraint on $\bm{y}$ such that
 \begin{align}
    p^x_{on} \cdot F^x_m = p^x_{on} \cdot \prod_{N_j \in \bm{y}}(1-p_{d,x,j}) \leq F_I.
    \label{eq.yconstraint}
 \end{align}

%\subsection{Channel Model}
Given the signal transmission power $P_j$ of $s_j$, the noise power $\sigma_x^2$ over $C_x$, and the average channel gain $h_{j,i,x}^2$ of the link between $j$ and its destination node $i$ over $C_x$, the transmission rate $R_{j,i,x}$ from $j$ to $i$ can be given as~\cite{jsac2008spectrumsensing}:
\begin{align}
   R_{j,i,x} = B_x \log {\left ( 1 + h_{j,i,x}^2 \dfrac{P_j}{\sigma_x^2}\right )}.
      \label{eq.transrate}
\end{align}
We consider that data transmission over each licensed channel $C_x$ is error-free with the available channel capacity in Eq.~(\ref{eq.transrate}).  

During the intra-cluster data transmission, the transmission power of each sensor node is fixed to avoid co-channel interference among neighboring clusters~\cite{liang2011twc}. The inter-cluster data transmission is also performed in TDMA, but CHs can adjust their transmission power for inter-cluster transmission when accessing a licensed channel. However, we assume that CHs do not adjust their power when they transmit data over $C_0$, to avoid potential interference to other applications operating on this license-free channel~\cite{interferencelicensefree}. The determination of the transmission power over the default license-free channel can be referred to existing solutions~\cite{powerallocation_2008tvt, powerallocation_2012jsac}, which is out of the scope of this paper. 
%Although Eq.~(\ref{eq.transrate}) is the maximum rate that data can be transmitted over $C_x$ according to the Shannon-Hartley theorem~\cite{shanon2014TIT}, we still consider the bit error rate of the information over $C_x$. If sensor $j$ transmits data using power $P_j$ over channel $C_x$, where the channel gain between sensor $j$ and $i$ is $h_{j,i,x}$, the bit error rate can be given by $er_{j,i,x}=\textit{Q}(\frac{P_j*h_{j,i,x}^2}{\sigma_x^2})$. If a packet contains $\tau $ bits of data, the probability for retransmission can be given by
 % \begin{eqnarray}
 %  \begin{split}
 %      {\lambda}_{j,i,x}=1-(1-er_{j,i,x})^\tau.
 %    \end{split}
 %    \label{eq.packetlossrate}
 % \end{eqnarray}

\subsection{Energy Consumption Model}
The energy consumption of sensor nodes mainly includes four parts: the energy consumption for spectrum sensing, spectrum switching, data transmission and reception. For each sensor node, we use $e_s$ to denote the energy consumption for sensing a licensed channel, which is fixed and the same for different channels. Meanwhile, sensor nodes need to consume energy to configure the radio and switch to a new channel. Therefore, we use $e_w$ to denote the energy consumption that a sensor node consumes for channel switching. For $s_j$, the data transmission energy consumption $E_{j,t}$ is based on the classic energy model~\cite{shu2006joint}, i.e., $E_{j,t} = (P_j + P_{j,c}) \cdot t_{j,x}$,
where $t_{j,x}$ is the data transmission time, $P_j$ is the transmission power and $P_{j,c}$ is the circuit power at $s_j$. Following a similar model in \cite{cui2005energy}, $P_{j,c}$ can be calculated as $P_{j,c} = \alpha_j + (\dfrac{1}{\eta} - 1)\cdot P_j$, where $\alpha_j$ is a transmission-power-independent component that accounts for the power consumed by the circuit, and $\eta$ is the power amplifier efficiency. Physically, $\eta$ is determined by the drain efficiency of the RF power amplifier and the modulation scheme~\cite{cui2005energy,shu2006joint}. Therefore, we have the energy consumption of data transmission at $s_j$ is
\begin{align}
\label{energymodel}
E_{j,t} = \dfrac{1}{\eta}\cdot P_j\cdot t_{j,x} + \alpha_j \cdot t_{j,x} = \dfrac{1}{\eta}(P_j + \alpha_{c,j})\cdot t_{j,x},
\end{align}
where $ \alpha_{c,j} = \eta \cdot \alpha_j$ is defined as the equivalent circuit power consumption for data transmission. The energy consumption for data receiving is related to the data that a sensor node receives~\cite{ju2014ksii}. If $s_j$ receives $l$ bits data, the energy consumption is $E_{j,r} = e_c \cdot l$, where $e_c$ is the circuit power for data receiving.

\subsection{Problem Statement}
During each data transmission period, $s_j$ generates $A_j$ sensed data to report to the sink. Since data transmission is independent among different periods, our objective is to efficiently transmit $A = \sum_{s_j \in \mathcal{N}}A_j$ data to the sink within a data transmission period, by determining the channel sensing and accessing decision according to the channel condition of $C_0$. As an indicator of the time-varied channel condition, the packet loss rate of $C_0$ is measured by each pair CM-CH and CH-Sink at the beginning of each transmission period, which is assumed to be stable in a data transmission period but may vary over different periods.

According to the network model, the data transmission consists of two phases: intra-cluster data transmission and inter-cluster data transmission. Therefore, we focus on reducing the energy consumption during the two phases, respectively. % To accurately determine the channel sensing and accessing decisions, we comprehensively consider the energy consumption in channel sensing and switching, the detection probability of a primary channel, and the PU protection. 
 Specifically, we aim to address the following two issues. 

(1) During the intra-cluster data transmission, each cluster $L_i$ should determine whether to sense and access a licensed channel according to the packet loss rate of $C_0$. When $L_i$ decides to sense and access a license channel, %to address the dynamic availability of licensed channels, 
the channel sensing and accessing sequence should be determined for $L_i$ to minimize the energy consumption of intra-cluster data transmission in a probabilistic way.  %When $L_i$ accesses a licensed channel $C_x$, the transmission time should be allocated among the sensor nodes of $L_i$, to minimize the energy consumption of intra-cluster data transmission under the consideration of PU protection. 

(2) During the inter-cluster data transmission, the channel sensing and accessing decision should also be carefully determined for potential energy consumption reduction. Since CHs can adjust their transmission power when accessing a licensed channel, the transmission power control and dynamic channel accessing should be jointly considered to minimize the energy consumption of inter-cluster data transmission.

To ease the presentation, the key notations are listed in Table~\ref{table1}.
\begin{table}[!t]
    \caption{The Key Notations}
    \centering
    \small
    \begin{tabular}{p{1.1cm}|p{6.9cm}}
         \hline
         \hline %\vspace{0.1cm}
         \textbf{Notation} & \textbf{Definition} \\
         \hline
         \hline
         $\mathcal{N}$ & Set of sensor nodes $\{s_1, s_2, ..., s_n\}$\\
         $\mathcal{L}$ & Set of clusters $\{L_1, L_2, ..., L_m\}$\\
         $\mathcal{C}$ & Set of licensed channels $\{C_1, C_2, ..., C_k\}$\\
         $H_i$ & Cluster head of $L_i$ \\
         $\mathcal{N}_i$ & Set of sensor nodes in $L_i$\\
         $C_0$ & The default license-free channel \\
         $B_x$ & Bandwidth of channel $C_x$ \\
         $e_s$ & Energy consumption for sensing a channel \\
         $e_w$ & Energy consumption for channel switching \\
         $R_{j,i,x}$ & Data transmission rate from $s_j$ to $H_i$ over $C_x$ \\
         $h_{j,i,x}^2$ & Average channel gain between $s_j$ and $H_i$ over $C_x$\\
         $\sigma_{x}^2$ & Average noise power over $C_x$\\
         $\lambda_{j,i,0}$ & Packet loss rate between $s_j$ and $H_i$ over $C_0$\\
         $\bm{y}$ & A fixed number of sensor nodes chosen for cooperative channel sensing\\
         $P_j$ & Transmission power of $s_j$ over $C_0$ during intra-cluster data transmission\\
         $A_j$ & $s_j$'s sensed data needed to be transmitted during the intra-cluster data transmission\\
         $E_{1,0}(i)$ & Energy consumption of $L_i$ by performing intra-cluster data transmission over $C_0$ \\
         $p_r$ & PU protection requirement \\
         $T_x$ & The determined maximum channel available time of $C_x$ under the required $p_r$ \\
         $t_{j,x}$ & Allocated transmission time of $s_j$ over an accessed licensed channel $C_x$ \\
         $E_{1,x}(i)$ & Energy consumption of $L_i$ by performing intra-cluster data transmission over $C_x$ \\
         $\overline{E_{1,x}}(i)$ & Expected energy consumption of $L_i$ by performing intra-cluster data transmission over $C_x$ \\
         $A_i$ & Aggregated data of $H_i$ needed to be transmitted during inter-cluster data transmission \\
         $E_{2,0}$ & Energy consumption of inter-cluster data transmission over $C_0$ \\
         $E_{2,x}$ &  Energy consumption of inter-cluster data transmission over $C_x$\\
         $\overline{E_{2,x}}$ & Expected energy consumption of inter-cluster data transmission over $C_x$ \\
         $E_{2,x}'$ & Equivalent energy consumption for optimizing $E_{2,x}$\\
         $P_{i,0}$ & Transmission power of $H_i$ over $C_0$ during inter-cluster data transmission\\
         $P_{i,x}$ & Allocated transmission power of $H_i$ over $C_x$ during inter-cluster data transmission\\
         $t_{i,x}$ & Allocated transmission time of $H_i$ over an accessed licensed channel $C_x$ \\
         $z^*$ & The optimal value of $z$\\
         \hline
    \end{tabular}
    \label{table1}
\end{table}

\section{Dynamic Channel Accessing for Intra-Cluster Data Transmission}
\label{sec4}
In this section, we analyze the energy consumption for intra-cluster data transmission, and determine the condition when a cluster should sense and switch to a licensed channel for potential energy saving. Then, based on the analysis, a dynamic channel sensing and accessing scheme is proposed to reduce the energy consumption of intra-cluster data transmission. %As intra-cluster data transmission is independent among different clusters, we focus on analyzing the intra-cluster data transmission of $L_i$ without loss of generality.

\subsection{Energy Consumption Analysis of Intra-cluster Data Transmission}
Since the objective is to reduce the energy consumption for intra-cluster data transmission in the clusters, by opportunistically sensing and accessing to the licensed channels when the packet loss rate of $C_0$ is high, the energy consumption should be calculated first if a cluster $L_i$ ($L_i \in \mathcal{L}$) gathers the intra-cluster data over $C_0$. \textbf{Proposition~\ref{prop1}} provides the energy consumption information for the clusters, given the measured packet loss rate of $C_0$ at each communication link. 
\begin{Prop}
\label{prop1}
For each cluster $L_i$, if the data amount of a cluster member $s_j$ $(s_j \in \mathcal{N}_i)$ is $A_j$, and the packet loss rate between $s_j$ and the cluster head $H_i$ over $C_0$ is $\lambda_{j,i,0}$, the energy consumption for intra-cluster data transmission is 
\begin{align}
\label{eq.ec}
E_{1,0}(i) = \sum_{s_j \in \mathcal{N}_i} \dfrac{A_j \cdot ER_{1,j}}{(1-\lambda_{j,i,0})},
\end{align}
where $ER_{1,j} = \dfrac{\eta \cdot R_{j,i,0} \cdot e_c + P_j + \alpha_{c,j}}{\eta \cdot R_{j,i,0}}$ means the energy consumption rate of $s_j$ for transmitting intra-cluster data, $R_{j,i,0} = B_0 \log {\left ( 1 + h_{j,i,0}^2 P_{j}/\sigma_0^2\right )}$ and $P_j$ is the transmission power of $s_j$.
\end{Prop}

\begin{proof}
For each $s_j \in \mathcal{N}_i$, it generates $A_i$ data to transmit during a data transmission period. Since the packet loss rate of $C_0$ is $\lambda_{j,i,0}$, the expected number of transmission attempts for each packet is $1 / (1-\lambda_{j,i,0})$. Therefore, the expected transmitted data is $A_j/(1-\lambda_{j,i,0})$. If the transmission power of $s_j$ is $P_j$, the data transmission time is $\dfrac{A_j}{(1-\lambda_{j,i,0}) R_{j,i,0}}$. Therefore, for all the sensor nodes in $L_i$, the energy consumption for data transmission is
\begin{align}
\label{eq.ect}
\nonumber
e_{1,t}(i) &= \sum_{s_j \in \mathcal{N}_i} \left [\dfrac{A_j \cdot (P_j + \alpha_{c,j})}{\eta \cdot (1-\lambda_{j,i,0}) \cdot R_{j,i,0}} \right ]\\
&=\sum_{s_j \in \mathcal{N}_i} \left [\dfrac{A_j \cdot (P_j + \alpha_{c,j})}{\eta (1-\lambda_{j,i,0}) B_0 \log {\left ( 1 + h_{j,i,0}^2 P_j/\sigma_0^2\right )}} \right ].
\end{align}

Additionally, the energy consumption for receiving the sensed data is 
\begin{align}
\label{eq.ecr}
e_{1,r}(i) = \sum_{s_j \in \mathcal{N}_i} \left (A_j \cdot \dfrac{1}{1-\lambda_{j,i,0}} \cdot e_c \right ).
\end{align}

Therefore, the total energy consumption of intra-cluster data transmission over $C_0$ is $E_{1,0}(i) = e_{1,t}(i) + e_{1,r}(i)$, which can be transformed to Eq.~(\ref{eq.ec}). It completes the proof.
\end{proof}

\subsection{Optimized Transmission Time Allocation for Intra-cluster Data Transmission}
\label{sec.optimaltransmissiontime}

According to Eq.~(\ref{eq.ec}), the energy consumption for intra-cluster data transmission in $L_i$ grows sharply with the increasing packet loss rate of $C_0$. If we aim to access licensed channel $C_x$ to reduce the intra-cluster energy consumption in $L_i$, we should first address the problem: how to allocate the transmission time of CMs to minimize the energy consumption with the consideration of PU protection. In this section, we focus on determining the optimized energy consumption if $L_i$ accesses $C_x$ for data transmission.

When $L_i$ accesses to $C_x$, the channel available duration (CAD) of $C_x$, denoted by $T_x$, is limited to control the interference probability to PUs, due to the fact that PUs may return at any time point and cause an interference with a certain probability. A longer CAD indicates a higher interference probability to PUs. We define $p_r$ as the PU protection requirement, which means the interference probability to PU during $T_x$ should be no larger than $p_r$. According to the cognitive radio model, the PU traffic is an independent and identically distributed ON/OFF process, with $v_x$ as the mean idle time. Thus, if $C_x$ is accessed for $T_x$, the interference probability of $C_x$ is $1-e^{-v_x \cdot T_x}$. Meanwhile, the probability that $C_x$ is idle and detected as idle is $p_{off}^x \cdot (1-F_f^x)$. Therefore, the interference probability during $T_x$ is $p_{off}^x \cdot (1-F_f^x) \cdot (1-e^{-v_x \cdot T_x})$, and the PU protection requirement is $p_{off}^x \cdot (1-F_f^x) \cdot (1-e^{-v_x \cdot T_x}) \le p_r$. Based on that, the maximum CAD of $C_x$ is
\begin{align}
    T_x = -\dfrac{1}{v_x}\ln\left (1-\dfrac{p_r}{p_{off}^x\cdot (1-F_f^x)} \right ).
\label{eq.availabletime}
\end{align}

If $T_x$ is large enough to guarantee the complement of the intra-cluster data transmission in $L_i$, all the data of CMs in $L_i$ can be transmitted over $C_x$. Otherwise, $T_x$ should be carefully allocated to the CMs of $L_i$ to minimize the energy consumption, since CMs have different amounts of sensed data and different transmission rates, both of which can directly impact the energy consumption of intra-cluster data transmission. In the following, we mathematically formulate the transmission time allocation problem as an optimization problem, which will be solved to minimize the energy consumption of intra-cluster data transmission.

For channel $C_x$ and cluster $L_i$, let $t_{j,x}$ be the allocated transmission time of $s_j$ ($s_j \in \mathcal{N}_i$) over $C_x$. Then, the energy consumption of $s_j$ for data transmission over $C_x$ is $e_{j,x} = \dfrac{1}{\eta}(P_j + \alpha_{c,j})\cdot t_{j,x}$. The residual data of $s_j$, if any, will be transmitted over $C_0$, with the amount of $ A_j - R_{j,i,x} \cdot t_{j,x}$. The associated energy consumption for transmitting the residual data over $C_0$ is $e_{j,0} = \dfrac{\left ( A_j - R_{j,i,x} t_{j,x}\right )\cdot ER_{1,j}}{1-\lambda_{j,i,0}}$. Let $E_{1,x}(i)$ be the total energy consumption for intra-cluster data transmission in $L_i$ by accessing $C_x$. Then, we have $E_{1,x}(i) = \sum_{s_j \in \mathcal{N}_i} \left ( e_{j,x} + e_{j,0} \right )$. There are also some constraints for the transmission time allocation of $T_x$. For each CM $s_j \in \mathcal{N}_i$, the successfully transmitted data of $s_j$ during the allocated time $t_{j,x}$ should be no larger than the generated data, which means
\begin{align}
\label{eq.rateconstraint}
R_{j,i,x} \cdot t_{j,x} \le A_i,~~~\forall s_j \in \mathcal{N}_i.
\end{align}
Meanwhile, the allocated transmission time $t_{j,x}$ of $s_j$ should be no less than 0 and the total allocated transmission time of $L_i$ should be no larger than $T_x$. Thus, we have
\begin{align}
\label{eq.timeconstraint}
\begin{cases}
  \sum_{s_j \in \mathcal{N}_i} t_{j,x} \le T_x,\\
  t_{j,x} \ge 0,~~~\forall s_j \in \mathcal{N}_i.
\end{cases}
\end{align}

We aim to determine the time allocation vector $\bm{t}_x = \{t_1, ..., t_{|\mathcal{N}_i|}\}$ to minimize the energy consumption of intra-cluster data transmission, which can be formulated as the following optimization problem:
\begin{align*}
\textbf{(TAP)~} \text{minimize~} &E_{1,x}(i) = \sum_{s_j \in \mathcal{N}_i} \left ( e_{j,x} + e_{j,0} \right )  \\\nonumber
{\rm s.t.}~&~(\ref{eq.rateconstraint})~\text{and}~(\ref{eq.timeconstraint}).
% \begin{cases}
%   B_x \log {\left ( 1 + \dfrac{h_{j,i,x}^2 P_{j}}{\sigma_x^2}\right )} \cdot t_{j,x} \le A_j,~~~\forall s_j \in \mathcal{N}_i,\\
%   \sum_{s_j \in \mathcal{N}_i} t_{j,x} \le T_x,\\
%   t_{j,x} \ge 0,~~~\forall s_j \in \mathcal{N}_i.
%   % \\
%   % \sum_{s_j \in \mathcal{N}_i} \left ( t_{j,x} + \dfrac{A_j - R_{j,i,x} \cdot t_{j,x}}{(1-\lambda_{j,i,0})\cdot R_{j,i,0}} \right) \le D_1.
% \end{cases}
\end{align*}

It can be seen that \textbf{(TAP)} is a classic linear programming problem. The well-known Simplex method can be directly applied to solve this problem~\cite{dantzig1998linear}. In the following, we use $\bm{t}_x^* = \{t_1^*, ..., t_{|\mathcal{N}_i|}^*\}$ and $E_{1,x}^*(i)$ to denote the optimal time allocation and energy consumption for intra-cluster data transmission by accessing $C_x$, respectively. %If we pay attention to the constraints of \textbf{(TAP)}, it can be found that the first three constraints can be met by adjusting $\bm{t}_x$, while the fourth constraint depends on the channel condition of $C_x$ and $C_0$. When there is no available solution to meet the fourth constraint (e.g., $(1-\lambda_{j,i,x})\cdot R_{j,i,x} < (1-\lambda_{j,i,0})\cdot R_{j,i,0},~\forall s_j \in \mathcal{N}_i $), we consider that $C_x$ is not suitable for $L_i$, and set $E_{1,x}^*(i) = \infty $. 

\subsection{Analysis of Channel Sensing and Switching Decision for Intra-cluster Data Transmission}
In this section, we focus on determining the condition when sensor nodes should sense and switch to a licensed channel for intra-cluster data transmission. By solving \textbf{(TAP)}, we can obtain the optimal energy consumption for transmitting intra-cluster data over $C_x$. However, due to the uncertain availability of $C_x$ and the energy consumption for channel sensing and switching, we can only obtain the expected energy consumption of intra-cluster data transmission by accessing $C_x$, if considering these two factors. According to the cognitive radio model, once $L_i$ decides to sense a licensed channel, a number of CMs $\bm{y}$ should be chosen to perform cooperative sensing to achieve better sensing performance. Here, $|\bm{y}|$ is a system parameter to meet the constraint of Eq.~(\ref{eq.yconstraint}), and we assume $|\bm{y}| \le \min_{C_i \in \mathcal{C}} {|\mathcal{N}_i|}$. 

Recall that, reducing the energy consumption of intra-cluster data transmission is the primary objective for channel sensing and switching. To determine if the energy consumption can be improved by sensing and switching to a licensed channel, we first define the expected accessible channel that is expectedly profitable for a cluster to sense and access.

\begin{Def}
\label{def1}
For cluster $L_i$, an expected accessible channel is a channel, by accessing which the expected energy consumption for intra-cluster data transmission can be reduced, taking account of the energy consumption for channel sensing and switching, as well as the idle detection probability of this channel by cooperative sensing.
\end{Def}

According to the definition, the following proposition determines the expected accessible channels for a specific cluster. 

\begin{Prop}
\label{prop2}
For channel $C_x$, given the detection probability $P_d^x$ and the false alarm probability $P_f^x$, the expected energy consumption of intra-cluster data transmission in $L_i$ by accessing $C_x$ is
\begin{equation}
\label{eq.expectedenergyconsumption}
\begin{aligned}
\overline{E_{1,x}}(i) &= E_{1,0}(i) + Y_{j,i,x} F_s^x t_{j,x}^* + 2|\mathcal{N}_i| e_w F_s^x + |\bm{y}| e_s,
\end{aligned}
\end{equation}
and $C_x$ is an expected accessible channel for $L_i$, if we have
\begin{equation}
\label{eq.accessiblechannel}
\begin{aligned}
Y_{j,i,x} \cdot F_s^x \cdot t_{j,x}^* + 2|\mathcal{N}_i| \cdot e_w \cdot F_s^x + |\bm{y}| \cdot e_s < 0,
\end{aligned}
\end{equation}
where $F_s^x = p_{off}^x \cdot (1 - F_f^x)$ and 
\begin{align*}
Y_{j,i,x} = \sum_{s_j \in \mathcal{N}_i} \Big ( \dfrac{(P_j+\alpha_{c,j})(1-\lambda_{j,i,0})-\eta ER_{1,j} R_{j,i,x}}{(1-\lambda_{j,i,0})\cdot \eta} \Big ).
\end{align*}
\end{Prop}

\begin{proof}
For channel $C_x$ with available time $T_x$, if it is used for intra-cluster data transmission in $L_i$, the optimal time allocation solution can be determined as $\bm{t}_x^* = \{t_1^*, ..., t_{|\mathcal{N}_i|}^*\}$ by solving \textbf{(TAP)}. Then, the optimal energy consumption for intra-cluster data transmission is
\begin{equation}
\label{eq.optimalx}
\begin{aligned}
E_{1,x}^*(i) = \sum_{s_j \in \mathcal{N}_i} \Big [ \dfrac{(P_j + \alpha_{c,j}) t_{j,x}^*}{\eta} 
+ \dfrac{\left ( A_j - R_{j,i,x} t_{j,x}^*\right ) ER_{1,j}}{1-\lambda_{j,i,0}} \Big ].
\end{aligned}
\end{equation}

If we consider the energy consumption for channel sensing and switching, the total energy consumption for using $C_x$ in intra-cluster data transmission is $E_{1,x}^*(i) + |\bm{y}| \cdot e_s + 2|\mathcal{N}_i| \cdot e_w$. Meanwhile, if $L_i$ decides to sense $C_x$, the probability that $C_x$ is detected as available is $F_s^x = p_{off}^x \cdot (1 - F_f^x)$, according to the cognitive radio model\footnote{When $C_x$ is detected as idle by cooperative sensing, there is also a probability that $C_x$ is not available at this time, which is $p^x_{on} \cdot F^x_m$. However, this probability is limited below $F_I$ by Eq.~(\ref{eq.yconstraint}), thus, we ignore it in the analysis of this work.}. It means that we have a probability $F_s^x$ to use $C_x$ and a probability $1-F_s^x$ to stay in channel $C_0$. Therefore, the expected energy consumption for sensing and switching to $C_x$ for intra-cluster data transmission is 
\begin{equation}
\label{eq.expectedec}
\begin{aligned}
\overline{E_{1,x}}(i) &= F_s^x \cdot \left (E_{1,x}^*(i) + |\bm{y}| \cdot e_s + 2|\mathcal{N}_i| \cdot e_w \right )\\
 &+ (1 - F_s^x) \cdot \left (E_{1,0}(i) + |\bm{y}| \cdot e_s \right )
\end{aligned}
\end{equation}

Substituting $E_{1,0}(i)$ and $E_{1,x}^*(i)$ according to Eq.~(\ref{eq.ec}) and~(\ref{eq.optimalx}), respectively, then Eq.~(\ref{eq.expectedenergyconsumption}) can be proved. If $C_x$ is an expected accessible channel for $L_i$, the expected energy consumption should be less than $E_{1,0}(i)$, i.e., $\overline{E_{1,x}}(i) < E_{1,0}(i)$. Substituting $E_{1,0}(i)$ and $\overline{E_{1,x}}(i)$ with Eq.~(\ref{eq.ec}) and~(\ref{eq.expectedec}), we can obtain Eq.~(\ref{eq.accessiblechannel}). 
\end{proof}

Based on \textbf{Proposition~\ref{prop2}}, we have the following corollary to determine the condition in which the cluster $L_i$ should sense licensed channels for intra-cluster data transmission.

\begin{Cor}
\label{cor1}
If there exists such channel $C_x \in \mathcal{C}$ that is an expected accessible channel of $L_i$, $L_i$ should sense new channels for intra-cluster data transmission. 
\end{Cor}

\begin{proof}
According to \textbf{Definition~\ref{def1}} and \textbf{Proposition~\ref{prop2}}, the expected energy consumption for intra-cluster data transmission can be reduced in $L_i$ by sensing and switching to the channel $C_x$, if $C_x$ is an expected accessible channel of $L_i$. Therefore, if there exists such channel $C_x \in \mathcal{C}$ that can meet the constraint of Eq.~(\ref{eq.accessiblechannel}), $L_i$ should sense this licensed channel for the potential energy efficiency improvement.
\end{proof}

\subsection{Dynamic Channel Accessing for Intra-cluster Data Transmission}
\label{sec4.4}
In this section, we propose a dynamic channel sensing and accessing scheme for the intra-cluster data transmission of each cluster. %Moreover, a simple case study is provided to illustrate the key idea of the proposed scheme.

With \textbf{Corollary~\ref{cor1}}, each cluster $L_i$ can decide whether it should sense a licensed channel for intra-cluster data transmission according to the packet loss rate of the default channel $C_0$. However, if there exist a set of expected accessible channels $\mathcal{C}'$ ($\mathcal{C}' \in \mathcal{C}$) for $L_i$, the problem is which one is the most profitable to sense and access for intra-cluster data transmission. \textbf{Proposition~\ref{prop2}} indicates that the channel with the lowest expected energy consumption $\overline{E_{1,x}}(i)$ should be sensed first. However, $\overline{E_{1,x}}(i)$ is only an expected value and the availabilities of licensed channels are totally opportunistic, which means the expected accessible channels may be detected as unavailable through spectrum sensing. Therefore, we arrange the expected accessible channel set $C_x \in \mathcal{C}'$ according to the increasing order $\overline{E_{1,x}}(i)$, and $L_i$ senses the channels of $\mathcal{C}'$ one by one according to the order until detecting a channel as idle. Then, $L_i$ switches to this channel for intra-cluster data transmission. Specifically, we discuss the dynamic channel sensing and accessing for intra-cluster data transmission in $L_i$ in the following situations.
\begin{enumerate}[(i)]
\item If $\mathcal{C}' = \emptyset$, it means that there is no expected accessible channel for $L_i$. The cluster does not sense any licensed channel and uses $C_0$ for intra-cluster data transmission. 
\item If $\mathcal{C}' \neq \emptyset$ and all the channels of $\mathcal{C}'$ are sensed as unavailable, $L_i$ transmits the intra-cluster data over $C_0$.
\item If $\mathcal{C}' \neq \emptyset$ and $C_x$ ($C_x \in \mathcal{C}'$) is sensed as idle by $L_i$, $L_i$ switches to $C_x$ and transmits the intra-cluster data over $C_x$. If the intra-cluster data is not completed after $T_x$, the channel sensing and accessing decision should be performed again. For each CM $s_j \in \mathcal{N}_i$, we denote the residual data of $s_j$ as $A_j'$. Then, we use $A_j'$ in \textbf{Propositions~\ref{prop1}} and~\textbf{\ref{prop2}} to determine the set of expected accessible channels $\mathcal{C}'$, and repeat the channel sensing and accessing according the three situations until the intra-cluster data transmission is finished in $L_i$.
\end{enumerate}

Based on the discussion above, Fig.~\ref{fig.decisionflow} shows a flow chart to illustrate the procedures. \textbf{Algorithm~\ref{algm1}} presents the main idea of the dynamic channel sensing and accessing scheme for intra-cluster data transmission.
\begin{figure}
\centering
\includegraphics[width=0.45\textwidth]{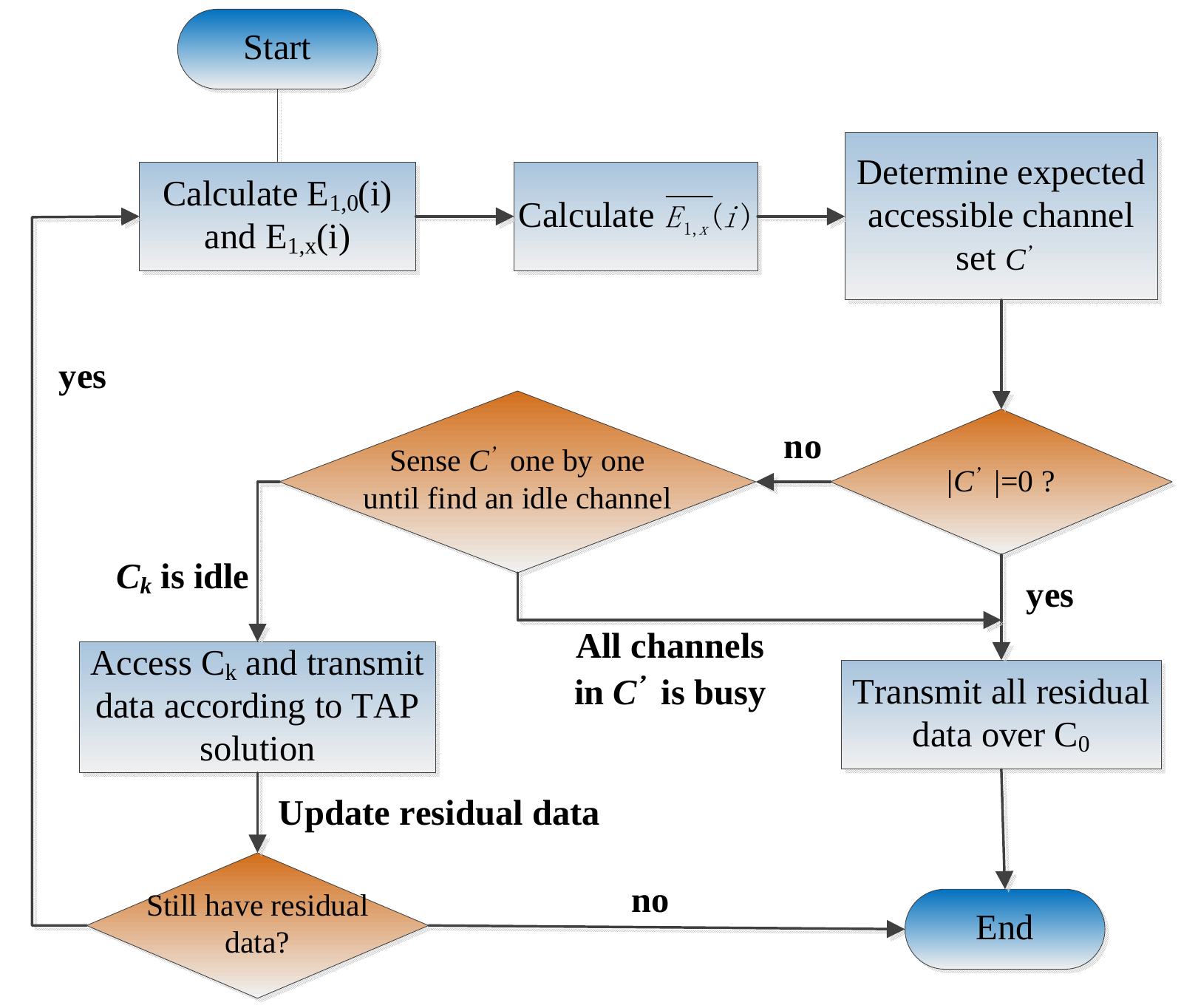}
\caption{Procedures of the dynamic channel sensing and accessing scheme.}
\label{fig.decisionflow}
\end{figure}
\begin{algorithm}[!ht]
    \caption{Dynamic Channel Sensing and Accessing for Intra-cluster Data Transmission.}
    \begin{algorithmic}[1]
        \REQUIRE
        For each $s_j$ and $C_x$, the sampling rate $sr_j$, the expected transmission rate $R_{j,i,x}$, packet loss rate $\lambda_{j,i,x}$, available transmission duration $T_x$, and other parameters in cognitive model and energy consumption model.
        \ENSURE
        Determining the channel sensing and accessing sequence for intra-cluster data transmission.
        \FORALL{$L_i \in \mathcal{L}$}
            \STATE{Calculate the energy consumption of intra-cluster data transmission $E_{1,0}(i)$ over $C_0$;}
            \FORALL{$C_x \in \mathcal{C}$}
                \STATE{Determine $E_{i,x}(i)$ and $\overline{E_{i,x}}(i)$ by solving \textbf{(TAP)} and according to \textbf{Proposition~\ref{prop2}}, respectively;}
            \ENDFOR
            \STATE{Determine the expected accessible channel set $\mathcal{C}'$ according to \textbf{Proposition~\ref{prop2}}, and reorder $\mathcal{C}'$ as $\overline{\mathcal{C}}'$ according to increasing order of $\overline{E_{i,x}}(i)$;}
            \STATE{$k~=~1$;}
            \WHILE{$k~\le~|\overline{\mathcal{C}}'|$}
                \STATE{Sense the $k$-th channel $C_k$ of $\overline{\mathcal{C}}'$;}
                \IF{$C_k$ is idle}
                    \STATE{Go to \textbf{step~18};}
                \ENDIF
                \STATE{$k~=~k~+~1$;}
            \ENDWHILE
            \IF{$~\overline{|\mathcal{C}}'|==0$ \textbf{~or~} $k>|\overline{\mathcal{C}}'|$}
                    \STATE{Transmit the residual intra-cluster data over the default channel $C_0$;}
            \ELSE
                    \STATE{Transmit the intra-cluster data over the channel $C_k$, and allocate the transmission time $t_{j,k}$ to each sensor node $N_j \in L_i$;}
                    \IF{The CAD of $C_k$ is expired \textbf{~and~} the intra-cluster data transmission of $L_i$ is not completed}
                        \STATE{Go to \textbf{step~2};}
                    \ENDIF
            \ENDIF
        \ENDFOR
    \end{algorithmic}
    \label{algm1}
\end{algorithm}

\section{Joint Power Allocation and Channel Accessing for Inter-Cluster Data Transmission}
\label{sec5}
After intra-cluster data transmission, CHs aggregate the received data, and then send the aggregated data to the sink. Based on the analysis of intra-cluster data transmission, in this section, we focus on the channel accessing problem to improve the energy efficiency of inter-cluster data transmission. 

\subsection{Analysis of Channel Sensing and Switching Decision for Inter-cluster Data Transmission}
By considering all the CHs and the sink as a cluster where CHs are CMs and the sink is the CH, the inter-cluster data transmission is similar to the intra-cluster data transmission. However, since there is no interference for TDMA-based inter-cluster transmission over licensed channels, CHs can adjust their transmission power to transmit their data to the sink when accessing to a licensed channel. As a result, the analysis of sensing and switching decision becomes different for inter-cluster data transmission.

Following the analytical path of intra-cluster data transmission, we first obtain the energy consumption of inter-cluster data transmission over $C_0$ in the following proposition. According to our model, CHs do not adjust their power when they transmit over $C_0$, to avoid potential interference to other applications transmitting over this license-free channel. Therefore, we have the following proposition.
\begin{Prop}
\label{prop3}
Given the data aggregation rate of $H_i$ $(L_i \in \mathcal{L})$ as $\psi_i$, the packet loss rate $\lambda_{i,s,0}$ between a cluster head $H_i$ and the sink over $C_0$, the energy consumption for inter-cluster data transmission over $C_0$ is 
% \begin{align*}
% \label{eq.ecforinter}
$E_{2,0} = \sum_{H_i \in \mathcal{L}} \dfrac{A_i \cdot ER_{2,i}}{(1-\lambda_{i,s,0})}$,
% \end{align*}
where $ER_{2,i} = \dfrac{\eta \cdot R_{i,s,0} \cdot e_c + P_{i,0} + \alpha_{c,j}}{\eta \cdot R_{i,s,0}}$ means the energy consumption rate of $H_i$ for transmitting inter-cluster data over $C_0$, $A_i = \sum_{s_j \in \mathcal{N}_i}A_j \cdot \psi_i$, $R_{i,s,0} = B_0 \log {\left ( 1 + h_{i,s,0}^2 P_{i,0}/\sigma_0^2\right )}$ and $P_{i,0}$ is the transmission power of $H_i$.
\end{Prop}
\begin{proof}
Similar to the proof of \textbf{Proposition~\ref{prop1}}. 
\end{proof} 

We then determine the minimized energy consumption of inter-cluster data transmission by accessing licensed channel $C_x$. Based on Eq.~(\ref{eq.availabletime}), we can calculate the CAD of $C_x$ as $T_x$. Note that, besides $T_x$, the transmission power of CHs can also be adjusted for the inter-cluster data transmission. For each $H_i$, let $P_{i,x}$ and $t_{i,x}$ denote the allocated transmission power and transmission time over $C_x$, respectively. The energy consumption of data transmission over $C_x$ is $e_{i,x} = \dfrac{1}{\eta}(P_{i,x} + \alpha_{c,i})\cdot t_{i,x}$, and the energy consumption of transmitting the residual data over $C_0$, if any, is $e_{0,x} = \Big [ A_i -  B_x \log\left(1+\dfrac{h_{i,s,x}^2 P_{i,x}}{\sigma_x^2} \right) t_{i,x}\Big ]\cdot ER_{2,i} \cdot \dfrac{1}{1-\lambda_{i,s,0}}$. To minimize the energy consumption, we can jointly determine the transmission power vector $\bm{P_x} = \{P_{1,x}, ..., P_{m,x}\}$ and transmission time vector $\bm{t_x}=\{t_{1,x}, ..., t_{m,x}\}$ of the CHs, which can be formulated as the following optimization problem: 
\begin{align*}
&\textbf{(PTAP)~} \text{minimize~} E_{2,x} = \sum_{H_i \in \mathcal{L}} \left ( e_{i,x} + e_{i,0} \right )  \\\nonumber
{\rm s.t.}~ &\begin{cases}
  B_x \cdot \log\left(1+\dfrac{h_{i,s,x}^2 P_{i,s}}{\sigma_x^2} \right) \cdot t_{i,x} \le A_i,~~~\forall H_i \in \mathcal{L},\\
  \sum_{H_i \in \mathcal{L}} t_{i,x} \le  T_x,\\
  t_{i,x} \ge 0,~~~\forall H_i \in \mathcal{L},\\
  0 \le P_i \le P_{max},~~~\forall H_i \in \mathcal{L},
\end{cases}
\end{align*}
where $P_{max}$ is the maximum power of CHs.

Since $P_{i,x}$ and $t_{i,x}$ are two continuous decision variable for each $H_i \in \mathcal{L}$, \textbf{(PTAP)} can be proved as a biconvex optimization problem. The analysis for the solution of \textbf{(PTAP)} will be discussed in the following subsection. Let $E_{2,x}^*$ denote the optimal energy consumption, and $\bm{P_x}^* = \{P_{1,x}^*, ..., P_{m,x}^*\}$ and $\bm{t_x}^*=\{t_{1,x}^*, ..., t_{m,x}^*\}$  denote the optimal allocated transmission power and time, respectively. Then, we can calculate the expected energy consumption by accessing $C_x$ and determine the expected accessible channel for inter-cluster data transmission with the following proposition.
\begin{Prop}
\label{prop4}
For channel $C_x$, given the detection probability $P_d^x$ and the false alarm probability $P_f^x$, the expected energy consumption of inter-cluster data transmission by accessing $C_x$ is
\begin{equation}
\label{eq.expectedenergyconsumptionforinter}
\begin{aligned}
\overline{E_{2,x}} &= E_{2,0} + Y_{i,s,x} F_s^x t_{i,x}^* + 2 m e_w F_s^x + |\bm{y}| e_s,
\end{aligned}
\end{equation}
and $C_x$ is an expected accessible channel for inter-cluster data transmission, if we have
\begin{equation}
\label{eq.accessiblechannelforinter}
\begin{aligned}
Y_{i,s,x} \cdot F_s^x \cdot t_{i,x}^* + 2\cdot m \cdot e_w \cdot F_s^x + |\bm{y}| \cdot e_s < 0,
\end{aligned}
\end{equation}
where $F_s^x = P_{off}^x \cdot (1 - P_f^x)$ and 
\begin{align*}
Y_{i,s,x} = \sum_{s_j \in \mathcal{N}_i} \Big ( \dfrac{(P_{i,x}^*+\alpha_{c,i})(1-\lambda_{i,s,0})-\eta \cdot ER_{2,j} R_{i,s,x}}{(1-\lambda_{i,s,0})\cdot \eta} \Big ).
\end{align*}
\end{Prop}

Based on \textbf{Proposition~\ref{prop4}}, the following corollary provides the condition when CHs should sense licensed channels for inter-cluster data transmission.
\begin{Cor}
\label{cor2}
If there exists such channel $C_x \in \mathcal{C}$ that can be an expected accessible channel for inter-cluster data transmission, CHs should sense licensed channels to transmit inter-cluster data to the sink. 
\end{Cor}

The proof to \textbf{Proposition~\ref{prop4}} and \textbf{Corollary~\ref{cor2}} are omitted, since they are similar to the proof of \textbf{Proposition~\ref{prop2}} and \textbf{Corollary~\ref{cor1}}.

\subsection{Joint Transmission Power and Time Allocation for Inter-cluster Data Transmission}
In this subsection, we aim to solve the joint transmission power and time allocation problem (i.e., \textbf{(PTAP)}) for minimizing the energy consumption of inter-cluster data transmission.

We first expand the objective function of \textbf{(PTAP)} as 
\begin{align}
\label{eq.ecforinter}
\nonumber
E_{2,x} &= \sum_{H_i \in \mathcal{L}} \dfrac{A_i \cdot ER_{2,i}}{1-\lambda_{i,s,0}} + \sum_{H_i \in \mathcal{L}} \dfrac{(P_{i,x}+\alpha_{c,i}) \cdot t_{i,x}}{\eta}\\
&~~~-\sum_{H_i \in \mathcal{L}} \dfrac{ B_x\log\left(1+h_{i,s,x}^2 P_{i,x}/\sigma_x^2 \right) t_{i,x} ER_{2,i}}{1-\lambda_{i,s,0}}.
\end{align}
Since $\sum\limits_{H_i \in \mathcal{L}} \dfrac{A_i \cdot ER_{2,i}}{1-\lambda_{i,s,0}} $ is independent with the decision variables, \textbf{(PTAP)} is equivalent to minimizing the residual two parts of $E_{2,x}$. Let $W_i \eqdef \dfrac{B_x \cdot ER_{2,i}}{1-\lambda_{i,s,0}}$ and $E_{2,x}' \eqdef \sum\limits_{H_i \in \mathcal{L}} \dfrac{(P_{i,x}+\alpha_{c,i}) \cdot t_{i,x}}{\eta} - \sum\limits_{H_i \in \mathcal{L}} \left ( W_i \log\left(1+ \dfrac{h_{i,s,x}^2 P_{i,x}}{\sigma_x^2} \right) t_{i,x} \right)$, the equivalent problem of \textbf{(PTAP)} can be given as follows,
\begin{align*}
&\textbf{(PTAP-E)~} \text{~minimize~~} E_{2,x}' \\\nonumber
{\rm s.t.}~ & {\rm ~the~same~constraints~as~\textbf{(PTAP)}}.
% {\rm s.t.}~ &\begin{cases}
%   B_x \cdot \log\left(1+\dfrac{h_{i,s,x}^2 P_{i,x}}{\sigma_x^2} \right) \cdot t_{i,x} \le A_i,~~~\forall H_i \in \mathcal{L},\\
%   \sum_{H_i \in \mathcal{L}} t_{i,x} \le  T_x,\\
%   t_{i,x} \ge 0,~~~\forall H_i \in \mathcal{L},\\
%   0 \le P_{i,x} \le P_{max},~~~\forall H_i \in \mathcal{L}.
% \end{cases}
\end{align*}

In the following, we focus on solving \textbf{(PTAP-E)} instead of \textbf{(PTAP)}. The main idea of solving the biconvex problem is to decouple the joint optimization problem into two sequential sub-problems. It can be achieved by first determining the optimal transmission power for a given transmission time $\bm{t_x}$ from the feasible set of transmission time. Then, using the determined optimal $\bm{P_x}$ to derive the optimal $\bm{t_x}$, which can be iteratively used to determine the optimal transmission power. With sufficient iteration, we can obtain the optimal energy consumption. The detailed proof of this decoupling approach is provided in~\cite{de1994block,shu2006joint}. Taking advantage of this property, the solution of \textbf{(PTAP-E)} can be determined as follows.

\subsubsection{\textbf{Sub-Problem 1 - Optimization of Transmission Power $\bm{P_{x}^*}$ Under Given $\bm{t_{x}}$ }} We first calculate the optimal power allocation vector $\bm{P_x}^*$, when the allocated transmission time vector $\bm{t}$ is fixed with $t_{i,x} \ge 0$ and $\sum_{H_i \in \mathcal{L}} t_{i,x} \le T_x$. \textbf{(PTAP-E)} is equivalent to 
\begin{align*}
&\textbf{(PTAP-E1)~} \text{minimize~} E_{2,x}'(\bm{P_x}) \\\nonumber
{\rm s.t.}~ &\begin{cases}
  B_x \cdot \log\left(1+\dfrac{h_{i,s,x}^2 P_{i,x}}{\sigma_x^2} \right) \cdot t_{i,x} \le A_i,~~~\forall H_i \in \mathcal{L},\\
  0 \le P_{i,x} \le P_{max},~~~\forall H_i \in \mathcal{L}.
\end{cases}
\end{align*}
Obviously, \textbf{(PTAP-E1)} is a convex optimization problem, due to the convex objective function and convex feasible sets. Note that, the first constraint of \textbf{(PTAP-E1)} can be rewritten as a linear constraint $P_{i,x} \le \left(2^{\frac{A_i}{B_xt_{i,x}}} - 1\right)\cdot \sigma_x^2 / h_{i,s,x}^2$, because both $B_x$ and $t_{i,x}$ are no less than 0 and the logarithm function is monotonously increasing over the feasible set. Therefore, we have the following proposition.

\begin{Prop}
\label{prop5}
If the optimal solution to \textbf{(PTAP-E1)} exists, i.e., the feasible set of \textbf{(PTAP-E1)} is not empty, the optimal power allocation $\bm{P_x^*}$ is
\begin{align}
\label{eq.optimaltransmissionpower}
P_{i,x}^* = 
&\begin{cases}
  0,~~~\text{if~} P_{i,x}|_{\frac{\partial f}{\partial P_{i,x}} = 0} \le 0;\\
  P_{i,x}|_{\frac{\partial f}{\partial P_{i,x}} = 0},~~~\text{if~} 0 < P_{i,x}|_{\frac{\partial f}{\partial P_{i,x}} = 0} \le P_{i,x}^B;\\
  P_{i,x}^B,~~~\text{otherwise}.
\end{cases},
\end{align}
where $P_{i,x}^B = \min\Big\{\left(2^{\frac{A_i}{B_xt_{i,x}}} - 1\right)\cdot \sigma_x^2 / h_{i,s,x}^2, P_{max}\Big\}$ and $P_{i,x}|_{\frac{\partial f}{\partial P_{i,x}} = 0} = \dfrac{W_i \cdot \eta}{\ln2} - \dfrac{\sigma_x^2}{h_{i,s,x}^2}$.
\end{Prop}
\begin{proof}
Due to the convexity of \textbf{(PTAP-E1)}, the locally optimal solution is the globally optimal solution. Let $f(P_{i,x}) \eqdef \dfrac{(P_{i,x}+\alpha_{c,i}) \cdot t_{i,x}}{\eta} - W_i \log\left(1+ \dfrac{h_{i,s,x}^2 P_{i,x}}{\sigma_x^2} \right) t_{i,x}$. Its first-order partial derivate is
\begin{align}
\frac{\partial f}{\partial P_{i,x}} = \dfrac{t_{i,x}}{\eta} - \dfrac{W_i \cdot t_{i,x} \cdot h_{i,s,x}^2}{\ln2\left( \sigma_x^2 + h_{i,s,x}^2 \cdot P_{i,x}\right)}.
\end{align}

Let $\frac{\partial f}{\partial P_{i,x}} = 0$, we have $P_{i,x} = \dfrac{W_i \cdot \eta}{\ln2} - \dfrac{\sigma_x^2}{h_{i,s,x}^2}$. Here, we set $P_{i,x}|_{\frac{\partial f}{\partial P_{i,x}} = 0} \eqdef \dfrac{W_i \cdot \eta}{\ln2} - \dfrac{\sigma_x^2}{h_{i,s,x}^2}$. Since $\frac{\partial f}{\partial P_{i,x}} $ is monotonously increasing over the constraint set of $P_{i,x}$, $f(P_{i,x})$ decreases when $P_{i,x} \le P_{i,x}|_{\frac{\partial f}{\partial P_{i,x}} = 0}$, and the situation reverses when $P_{i,x} \ge P_{i,x}|_{\frac{\partial f}{\partial P_{i,x}} = 0}$. $f(P_{i,x})$ would achieve the maximum value at $P_{i,x}|_{\frac{\partial f}{\partial P_{i,x}} = 0}$.

Meanwhile, according to the constraints of $P_{i,x}$, the feasible set of $P_{i,x}$ is $P_{i,x} \le \left(2^{\frac{A_i}{B_xt_{i,x}}} - 1\right)\cdot \sigma_x^2 / h_{i,s,x}^2$ and $P_{i,x} \le P_{max}$. Let $P_{i,x}^B \eqdef \min\Big\{\left(2^{\frac{A_i}{B_xt_{i,x}}} - 1\right)\cdot \sigma_x^2 / h_{i,s,x}^2, P_{max}\Big\}$. Then, we have $P_{i,x}^* = 0$ if $P_{i,x}|_{\frac{\partial f}{\partial P_{i,x}} = 0} \le 0$; and $P_{i,x}^* = P_{i,x}|_{\frac{\partial f}{\partial P_{i,x}} = 0}$, if $0 < P_{i,x}|_{\frac{\partial f}{\partial P_{i,x}} = 0} \le P_{i,x}^B$; otherwise, $P_{i,x}^* = P_{i,x}^B$. It completes the proof. 
\end{proof}

\subsubsection{\textbf{Sub-Problem 2 - Optimization of Transmission Time $\bm{t_{x}^*}$ Under Given $\bm{P_{x}^*}$}}
Given any feasible transmission time vector $\bm{t_{x}}$, \textbf{Proposition~\ref{prop5}} presents the optimal transmission power vector in terms of $\bm{t_{x}}$ if such an optimal solution exists. In this sub-problem, we aim to determine the optimal $\bm{t_x^*}$ for given $\bm{P_x}$. The sub-problem 2 can be described as follows.
%If we substitute $P_i$ with the derived optimal transmission power $P_i^*$, \textbf{(PTAP-E)} becomes an optimization problem with only one vector of decision variables $\bm{t}$, which is to 
\begin{align*}
&\textbf{(PTAP-E2)~} \text{minimize~} E_{2,x}'(\bm{t_{x}})\\\nonumber
{\rm s.t.}~ &\begin{cases}
  B_x \cdot \log\left(1+\dfrac{h_{i,s,x}^2 P_{i,x}}{\sigma_x^2} \right) \cdot t_{i,x} \le A_i,~~~\forall H_i \in \mathcal{L},\\
  \sum_{H_i \in \mathcal{L}} t_{i,x} \le  T_x,\\
  t_{i,x} \ge 0,~~~\forall H_i \in \mathcal{L}.
\end{cases}
\end{align*}

Since $P_{i,x}$ is a fixed parameter in this sub-problem, \textbf{(PTAP-E2)} becomes a linear programming problem, similar to \textbf{(TAP)} of Section~\ref{sec.optimaltransmissiontime}. The Simplex algorithm can be applied to determine the optimal $\bm{t_{x}^*}$ and energy consumption~\cite{dantzig1998linear}.

Based on the analysis of the two sub-problems with \textbf{(PTAP-E)}, we focus on \textbf{(PTAP)} using the Alternative Convex Search method, which is a special case of the Block-Relaxation Methods~\cite{de1994block}. We illustrate the main idea of addressing this biconvex problem in Fig.~\ref{fig.iterationillustration}, and summarize the detailed procedures of our solution in \textbf{Algorithm~\ref{alg.biconvex}}.
\begin{figure}
\centering
\includegraphics[width=0.45\textwidth]{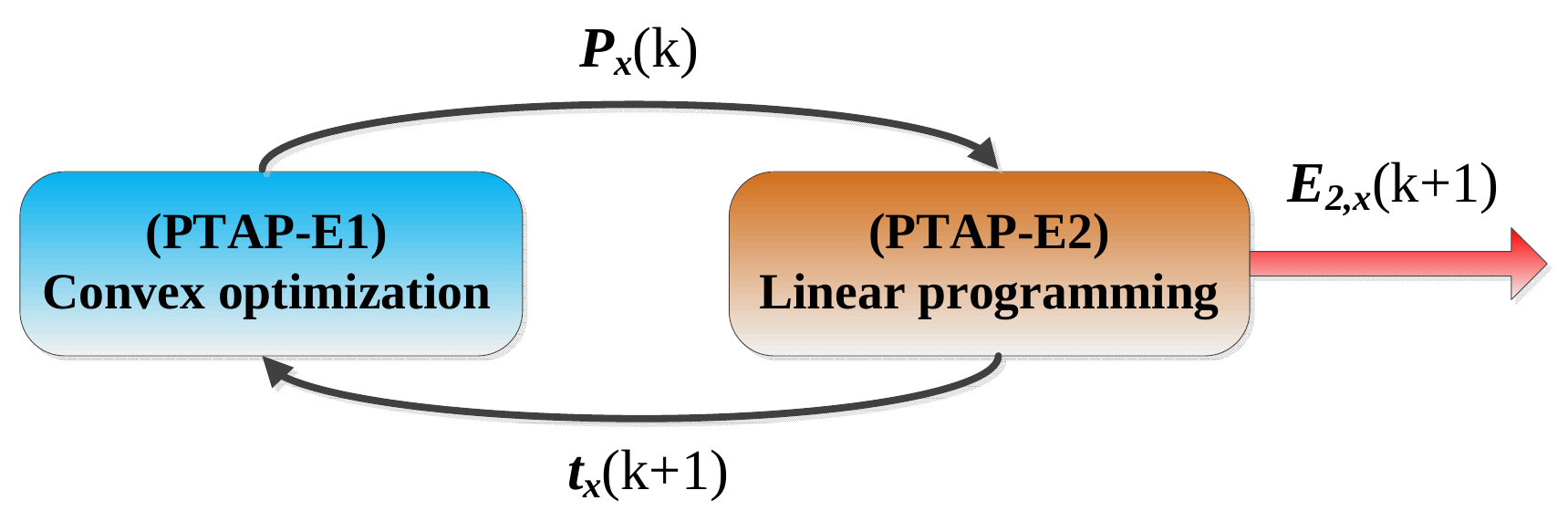}
\caption{Illustration of solving \textbf{(PTAP-E)}.}
\label{fig.iterationillustration}
\end{figure}
\begin{algorithm}[!ht]
    \caption{Alternative Convex Search based Algorithm for Solving \textbf{(PTAP)}.}
    \begin{algorithmic}[1]
        \REQUIRE
        The parameters of \textbf{(PTAP)}, convergence requirement $\omega$, and the maximum iteration number $\Lambda$.
        \ENSURE
        Determining the optimal $\bm{P_x^*}$ and $\bm{t_x^*}$, as well as the optimal energy consumption $E_{2,x}^*$.
            \STATE{Choose an arbitrary start point $\{\bm{P_x}(0), \bm{t_x}(0)\}$ from the feasible set of $\bm{P_x}$ and $\bm{t_x}$, and set $k=0$, $E_{2,x}(0) = 0$ ;}
            \REPEAT
            	\STATE{For given $\bm{t_x}(k)$, determine the optimal $\bm{P_x}(k+1)$ according to Eq.~(\ref{eq.optimaltransmissionpower});}
            	\STATE{For given $\bm{P_x}(k+1)$, determine the optimal $\bm{t_x}(k+1)$ and $E_{2,x}(k+1)$ by solving the linear programming \textbf{(PTAP-E2)};}
            	\STATE{$k = k+1$;}
            \UNTIL{$E_{2,x}(k) - E_{2,x}(k-1) \le \omega$ or $k \ge \Lambda$;}
            \RETURN{$\bm{P_x}(k)$, $\bm{t_x}(k)$, and $E_{2,x}(k) + \sum\limits_{H_i \in \mathcal{L}} \dfrac{A_i \cdot ER_{2,i}}{1-\lambda_{i,s,0}} $;}
    \end{algorithmic}
    \label{alg.biconvex}
\end{algorithm}

\subsection{Joint Power Allocation and Channel Accessing for Inter-cluster Data Transmission}
In this subsection, we propose a joint power allocation and channel accessing scheme for inter-cluster data transmission. Similar to the analysis in Section~\ref{sec4.4}, the channel sensing decision is made according to \textbf{Corollary~\ref{cor2}}. Moreover, the channel sensing and accessing sequence should follow the ordered expected accessible channel set with an increasing order of $\overline{E_{2,x}}$. For the three situations considered in Section~\ref{sec4.4}, they can be addressed during the inter-cluster data transmission with the same logic flow. However, as the transmission power of sensor nodes can be adjusted for different accessed channels, the channel accessing scheme during inter-cluster data transmission is combined with the power allocation scheme. We summarize the main idea of our joint power allocation and channel accessing scheme for inter-cluster data transmission in \textbf{Algorithm~\ref{algm2}}.  
\begin{algorithm}[!ht]
    \caption{Joint Power Allocation and Channel Accessing for Inter-cluster Data Transmission.}
    \begin{algorithmic}[1]
        \REQUIRE
        For each $H_i$, the aggregated data amount of $H_i$ and the packet loss rate $\lambda_{i,s,0}$ over $C_0$, and the parameters in cognitive model and energy consumption model.
        \ENSURE
        Determining the transmission power for cluster heads, the channel sensing and accessing sequence for inter-cluster data transmission.
            \STATE{Calculate the energy consumption $E_2$ over $C_0$ according to \textbf{Proposition~\ref{prop3}};}
            \FORALL{$C_x \in \mathcal{C}$}
                \STATE{Determine $E_{2,x}^*$ and $\overline{E_{2,x}}$ by solving \textbf{(PTAP)} and according to \textbf{Proposition~\ref{prop4}}, respectively;}
            \ENDFOR
            \STATE{Determine the expected accessible channel set $\mathcal{C}'$ according to \textbf{Proposition~\ref{prop4}}, and reorder $\mathcal{C}'$ as $\overline{\mathcal{C}}'$ according to increasing order of $\overline{E_{2,x}}$;}
            \STATE{$k~=~1$;}
            \WHILE{$k~\le~|\overline{\mathcal{C}}'|$}
                \STATE{Sense the $k$-th channel $C_k$ of $\overline{\mathcal{C}}'$;}
                \IF{$C_k$ is idle}
                    \STATE{Go to \textbf{step~17};}
                \ENDIF
                \STATE{$k~=~k~+~1$;}
            \ENDWHILE
            \IF{$~\overline{|\mathcal{C}}'|==0$ \textbf{~or~} $k>|\overline{\mathcal{C}}'|$}
                    \STATE{Transmit the residual inter-cluster data over the default channel $C_0$;}
            \ELSE
                    \STATE{Transmit the inter-cluster data over the channel $C_k$, and allocate the transmission time $t_{i,x}^*$ and adjust the transmission power to $P_{i,x}^*$ for each $H_i \in \mathcal{L}$, according to \textbf{Algorithm~\ref{alg.biconvex}};}
                    \IF{~The CAD of $C_k$ is expired \textbf{~and~} the inter-cluster data transmission is not completed}
                        \STATE{Go to \textbf{step~1};}
                    \ENDIF
            \ENDIF
    \end{algorithmic}
    \label{algm2}
\end{algorithm}

\section{Performance Evaluation}
\label{sec6}
We evaluate the performance of the proposed schemes by extensive simulations on OMNET++~\cite{ju2014jcsu, ju2014ksii}. We setup a network consisting of 200 sensor nodes forming 10 clusters. Sensor nodes are randomly deployed in a circular area with the network radius of 250 m, and the sink is located at the center. There are 15 licensed channels in the primary network, which can be sensed and accessed by the CRSN. All the channels including the default working channel $C_0$ are modeled as Rayleigh fading channels. For each channel $C_x$ ($C_x \in \mathcal{C}\cup \{C_0\}$), the noise spectral density is $10^{-14}$ W/Hz (i.e., $\sigma_x^2 = B_x \cdot 10^{-14}$W), and the channel gain between $s_i$ and $s_j$ is set as $h_{i,j,x}^2 = \gamma \cdot d_{i,j}^{-\mu}$, where $\gamma$ is an exponential random variable with mean value $1$, and $d_{i,j}$ is the distance between $s_i$ and $s_j$, and $\mu = 3$. Instead of setting the parameters of PU traffic on different licensed channels, we directly set the probability that PU is on as $p_{on}^x = 60\%$ and the channel available duration (CAD) as $T_x = \bm{N}$(100, 20) ms for each $C_x$, where $\bm{N}(a, b)$ means the normal distribution with mean value $a$ and variance $b$. The other parameters, if not specified in the simulation figures, are given in Table~\ref{table.parametersettings}.
\begin{table}[!t]
    \caption{Parameter Settings}
    \centering
    \small
    \begin{tabular}{p{5.8cm}|p{2.2cm}}
         \hline
         \hline %\vspace{0.1cm}
         \textbf{Parameter} & \textbf{Settings} \\
         \hline
         \hline
         CMs' power for intra-cluster data transmission $P_j$ & 20 mW\\
         CHs' power $P_{i,0}$ for inter-cluster data transmission over $C_0$ & 40 mW\\
         CHs' maximum adjustable power when accessing licensed channels $P_{max}$ & 200 mW\\
         Power amplifier efficiency $\eta$~\cite{shu2006joint} & 0.9\\
         Circuit power $\alpha_{c,j}$~\cite{shu2006joint} & 5 mW\\
         Energy consumption for data receiving $e_c$ & 5 nJ/bit\\
         Per-node energy consumption for sensing a licensed channel $e_s$~\cite{sensingcost_2010tvt} & $1.31 \times 10^{-4}$ J\\
         Per-node energy consumption for channel switching $e_w$~\cite{switchingcost_2013tvt} & $10^{-5}$ J\\
         Data amount of sensor node $A_j$ & $\bm{N}(5, 0.5)$ Kb\\
         Aggregation rate at CHs $\phi$ & $70\%$ \\
         Number of cooperative sensing nodes $|\bm{y}|$ & 3 \\
         Bandwidth of license-free channel $C_0$ & 1 MHz \\
         Bandwidth of licensed channel $C_x$ & $\bm{N}(2, 0.5)$ MHz \\
         False alarm probability $F_f^{x}$ & $5\%$ \\
         \hline
    \end{tabular}
    \label{table.parametersettings}
\end{table}

\subsection{Intra-cluster Data Transmission}
We evaluate the performance of the dynamic channel sensing and accessing scheme for intra-cluster data transmission in this subsection. Fig.~\ref{fig.1} shows the energy consumption of intra-cluster data transmission by accessing a specific licensed channel. In our proposed scheme, the CAD of the accessed channel is allocated to CMs according to the optimal solution of \textbf{(TAP)}. We compare our scheme to the average allocation scheme, in which the CAD of the accessed channel is equally allocated to the CMs with residual data. It can be seen that our scheme can achieve lower energy consumption than that of the average allocation scheme, when the CAD of the accessed channel is no large than 60 ms. After the CAD becomes larger than 60 ms, the energy consumption of two schemes converge to the same value. The reason is that a large CAD can guarantee that all the intra-cluster data are transmitted over the licensed channel, which leads to a minimum and stable energy consumption. 
\begin{figure}
\centering
\includegraphics[width=0.43\textwidth]{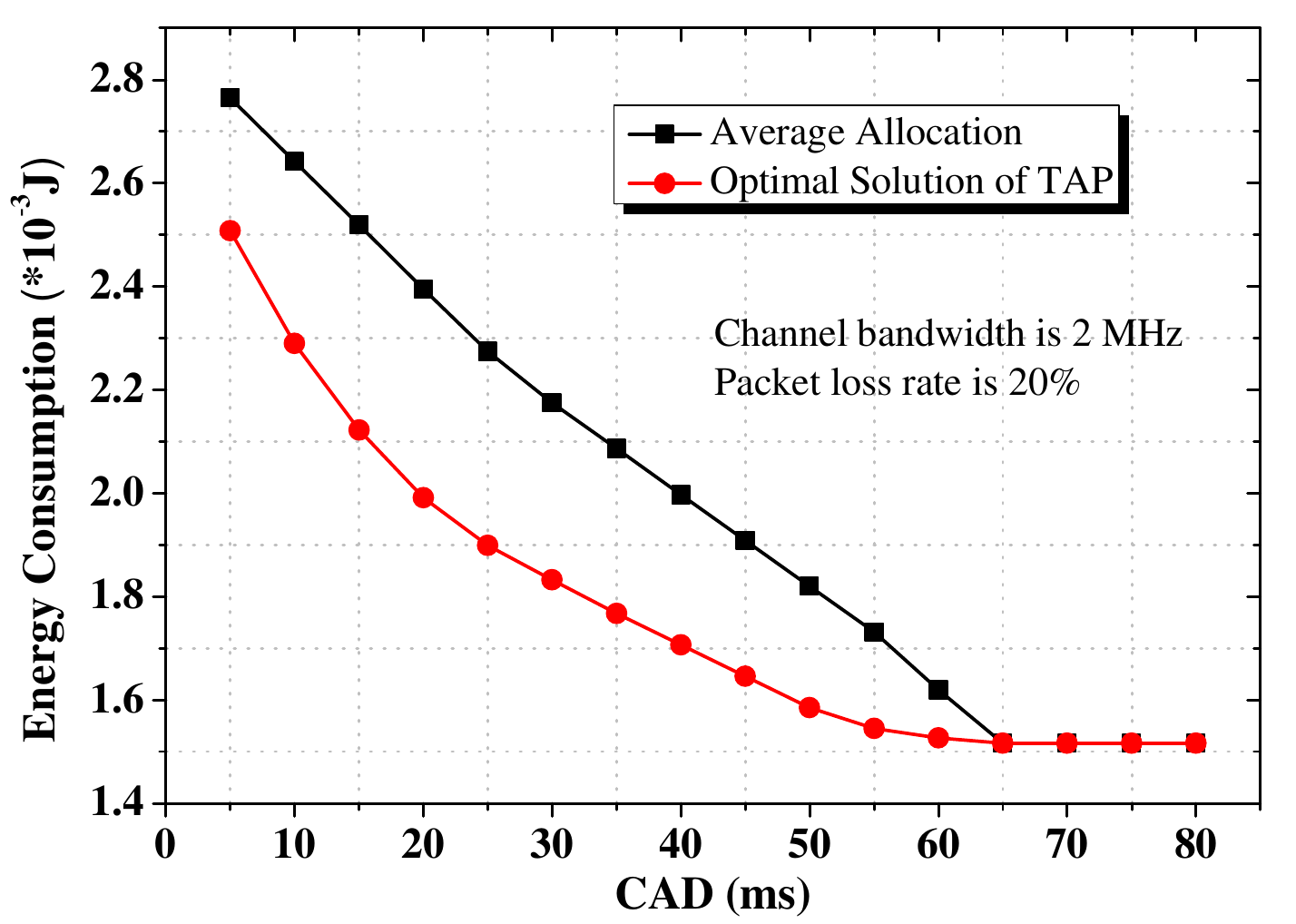}
\caption{Energy consumption comparison for intra-cluster data transmission by accessing a specific licensed channel.}
\label{fig.1}
\end{figure}

Fig.~\ref{fig.2} compares the energy consumption of intra-cluster data transmission under different packet loss rates over $C_0$. In the figure, the proposed scheme means \textbf{Algorithm~\ref{algm1}} and ASA refers to the scheme in which each cluster always senses and accesses a licensed channel for intra-cluster data transmission without considering the channel condition of $C_0$ and the energy consumption in channel sensing and switching. Moreover, different from the proposed scheme, the licensed channels with larger bandwidth are sensed and switched with a higher priority in ASA. It can be seen that energy consumption increases sharply with the increasing packet loss rate of $C_0$, if the cluster only uses $C_0$ for intra-cluster data transmission. The energy consumption of ASA is higher than that of only using $C_0$ when the packet loss rate is lower than $35\%$, and the situation reverses once the packet loss rate exceeds $35\%$. However, our proposed scheme has a much lower energy consumption than the others after the packet loss rate of $C_0$ exceeds $25\%$. Note that, the energy consumption of our proposed scheme is the same as that of only using $C_0$ when the packet loss rate of $C_0$ is smaller than $25\%$, because there is no expected accessible channel for intra-cluster data transmission.
\begin{figure}
\centering
\includegraphics[width=0.43\textwidth]{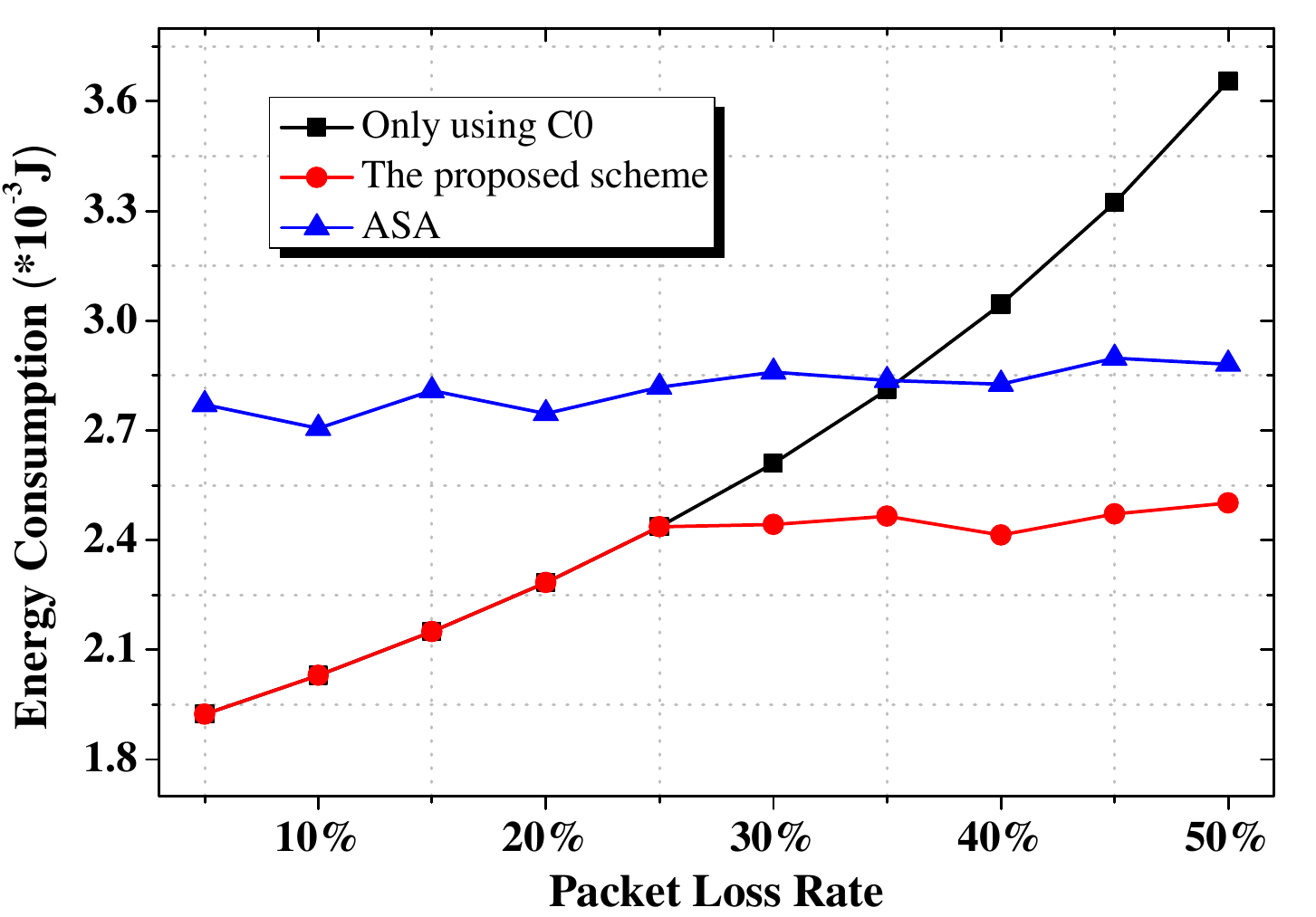}
\caption{Energy consumption comparison for intra-cluster data transmission under different packet loss rates.}
\label{fig.2}
\end{figure}

\subsection{Inter-cluster Data Transmission}
In this subsection, we aim to evaluate the performance of the joint power allocation and channel accessing scheme in inter-cluster data transmission. Fig.~\ref{fig.3} shows the convergence speed of the ACS based algorithm for solving \textbf{(PTAP)}, i.e., \textbf{Algorithm~\ref{alg.biconvex}}. It can be seen that the algorithm can converge (or find the optimal solution) within 6 iterations, which indicates the proposed algorithm is highly efficient and can be applied to resource-limited sensor networks. Fig.~\ref{fig.4} compares the energy consumption for inter-cluster data transmission under the proposed joint transmission power and time allocation scheme and the average allocation scheme. In the average allocation scheme, the CAD of the accessed channel is equally allocated to the CHs with residual data and CHs use the maximum power to transmit their data when using the accessed channel. In our scheme, the transmission time and power are allocated to CHs according to the optimal solution of \textbf{(PTAP)}. From Fig.~\ref{fig.4}, we can see that the average allocation scheme consumes much more energy than our proposed scheme under both scenarios of $P_{max} = 50$mW and $P_{max} = 200$mW. Meanwhile, our proposed scheme has lower energy consumption when it has a larger range of adjustable transmission power.  
\begin{figure}
\centering
\includegraphics[width=0.43\textwidth]{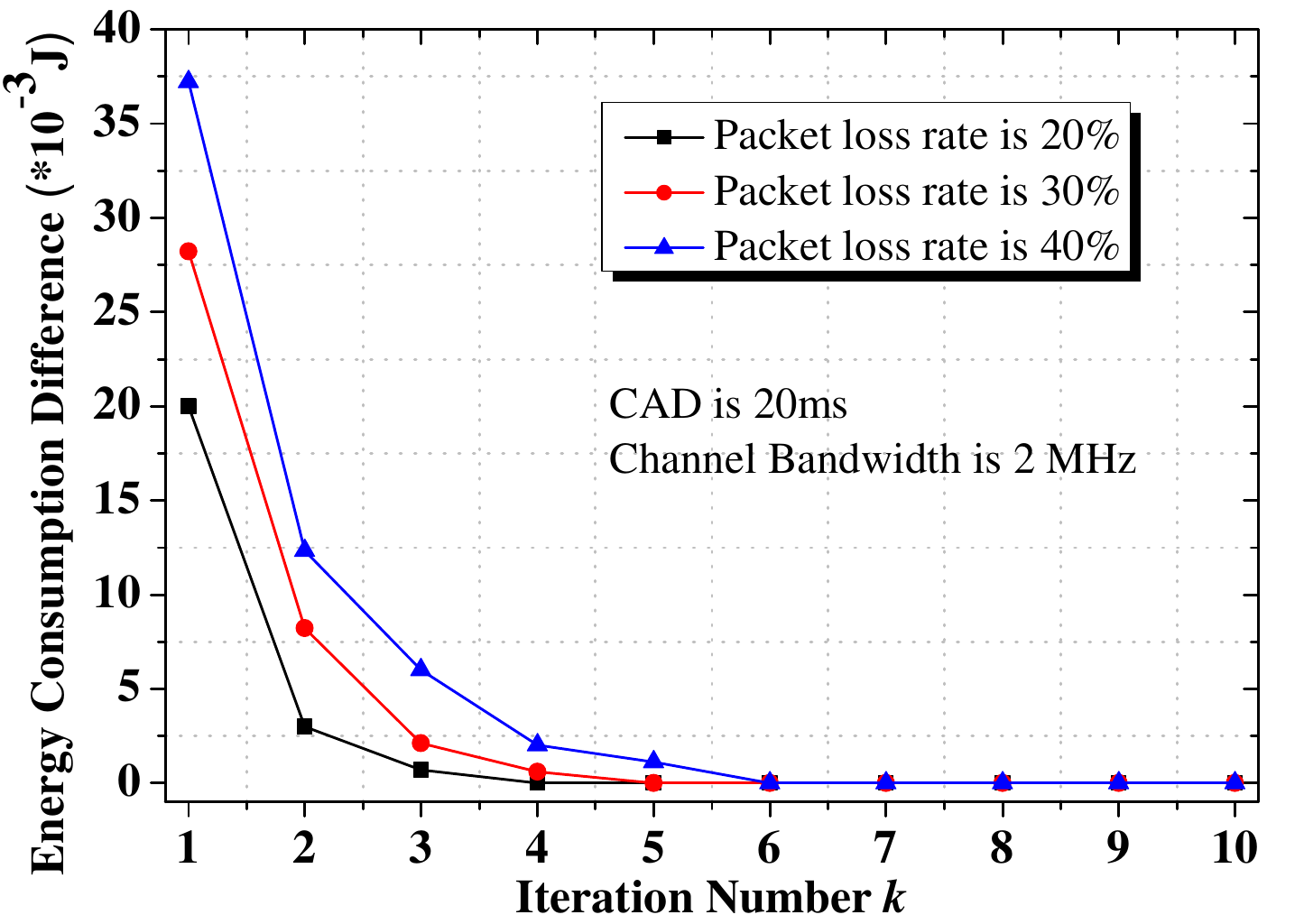}
\caption{Convergence speed of ACS based algorithm for solving \textbf{PTAP}.}
\label{fig.3}
\end{figure}
\begin{figure}
\centering
\includegraphics[width=0.43\textwidth]{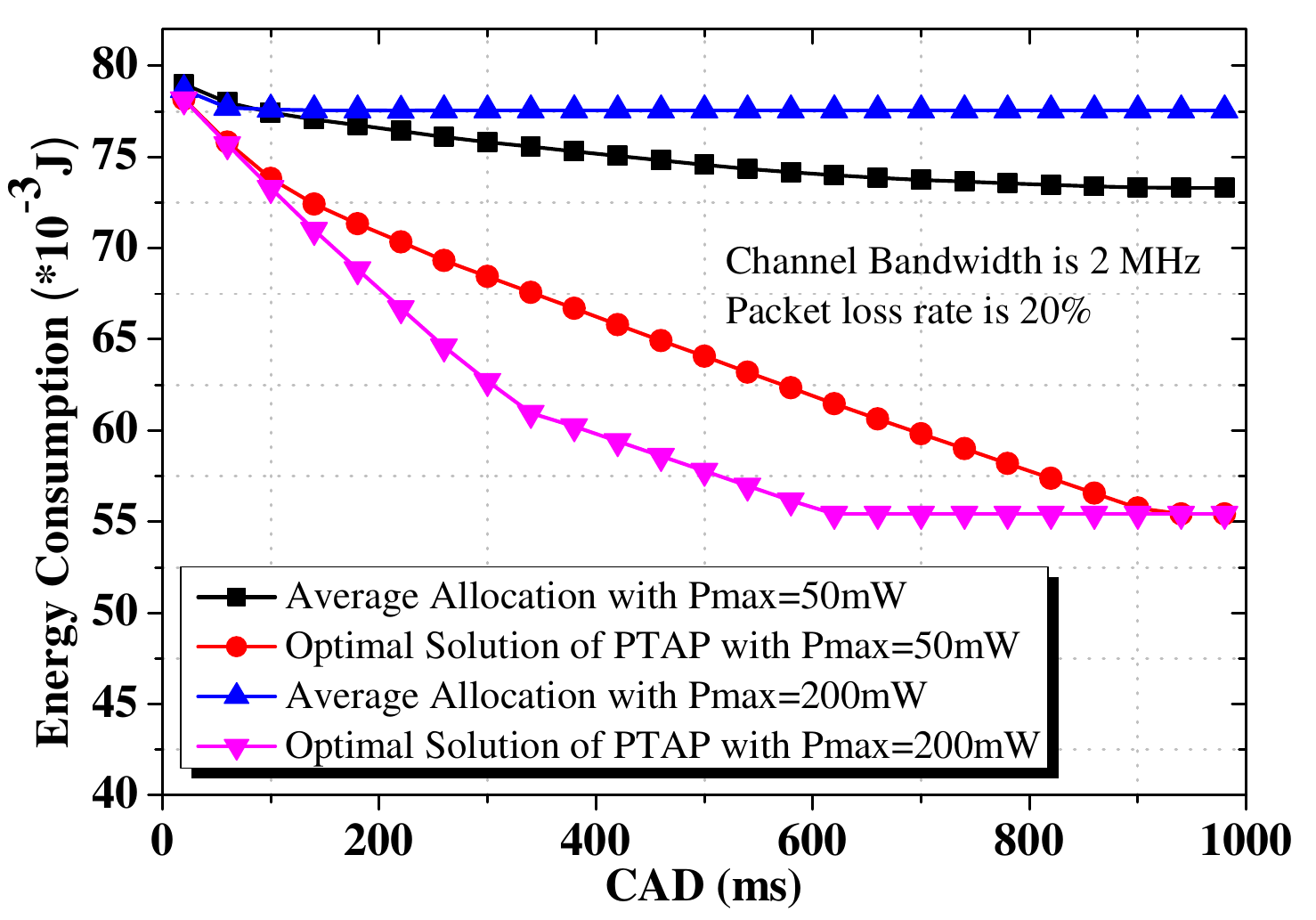}
\caption{Energy consumption comparison for inter-cluster data transmission by accessing a specific licensed channel.}
\label{fig.4}
\end{figure}

Fig.~\ref{fig.5} shows the energy consumption for inter-cluster data transmission with respect to different packet loss rates over $C_0$. Similar to the comparison in Fig.~\ref{fig.2}, our proposed scheme can achieve lower energy consumption than the others after the packet loss rate of $C_0$ is larger than $7\%$, and before that the energy consumption of our proposed scheme is the same as that of only using $C_0$. It indicates that the intra-cluster data transmission will be performed over $C_0$ in our proposed scheme when its packet loss rate is lower than $7\%$. Moreover, compared with the intra-cluster data transmission in Fig.~\ref{fig.2}, licensed channels are sensed and accessed at a lower packet loss rate over $C_0$. Because the heavy data traffic in inter-cluster data transmission can make the channel sensing and accessing profitable, even with a low packet loss rate over $C_0$.
\begin{figure}
\centering
\includegraphics[width=0.43\textwidth]{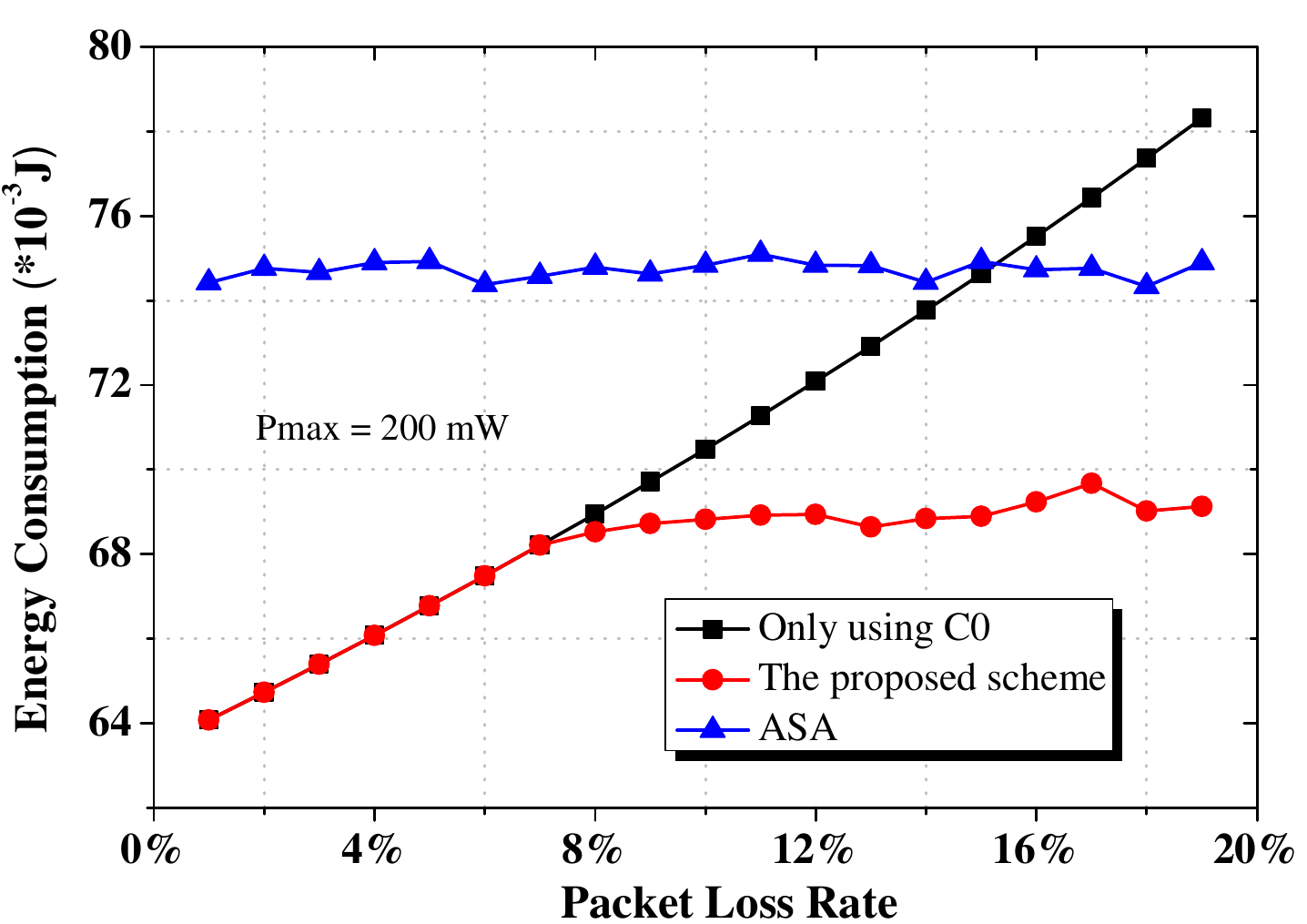}
\caption{Energy consumption comparison in inter-cluster data transmission under different packet loss rates.}
\label{fig.5}
\end{figure}

\subsection{Impacts of System Parameters}
Fig.~\ref{fig.6} shows the total energy consumption comparison under different amount of data traffic. With the increasing data amount transmitted by sensor nodes, the total energy consumption increases sharply if the CRSN uses $C_0$ for data transmission, while it only increases linearly under our proposed schemes. Moreover, higher data traffic indicates better energy consumption improvement. Fig.~\ref{fig.6} shows the total energy consumption comparison under different numbers of licensed channels. It can be seen that the total energy consumption in our proposed schemes decreases with the increasing number of licensed channels. 
\begin{figure}
\centering
\includegraphics[width=0.43\textwidth]{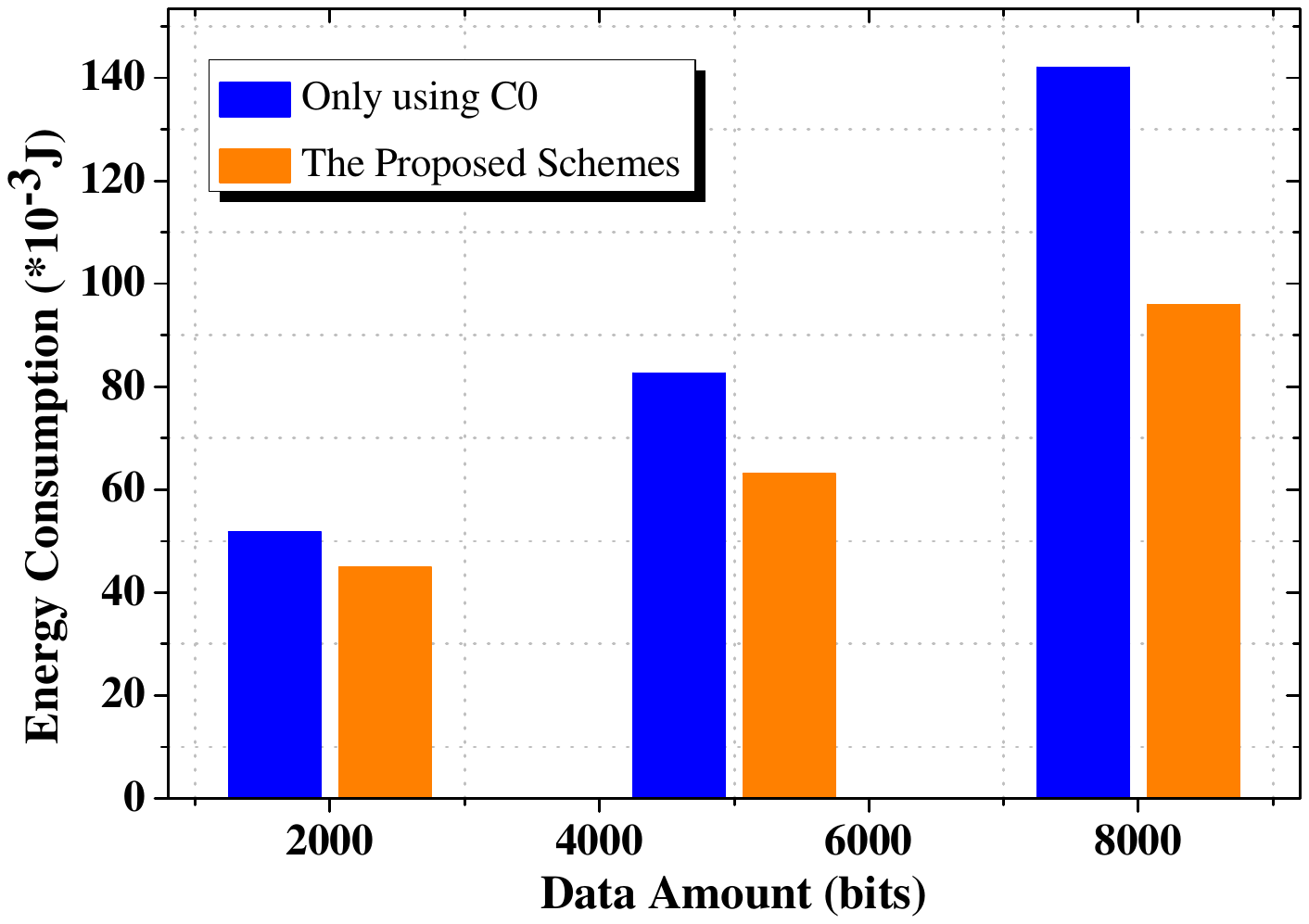}
\caption{Total energy consumption comparison under different amount of data traffic.}
\label{fig.6}
\end{figure}
\begin{figure}
\centering
\includegraphics[width=0.43\textwidth]{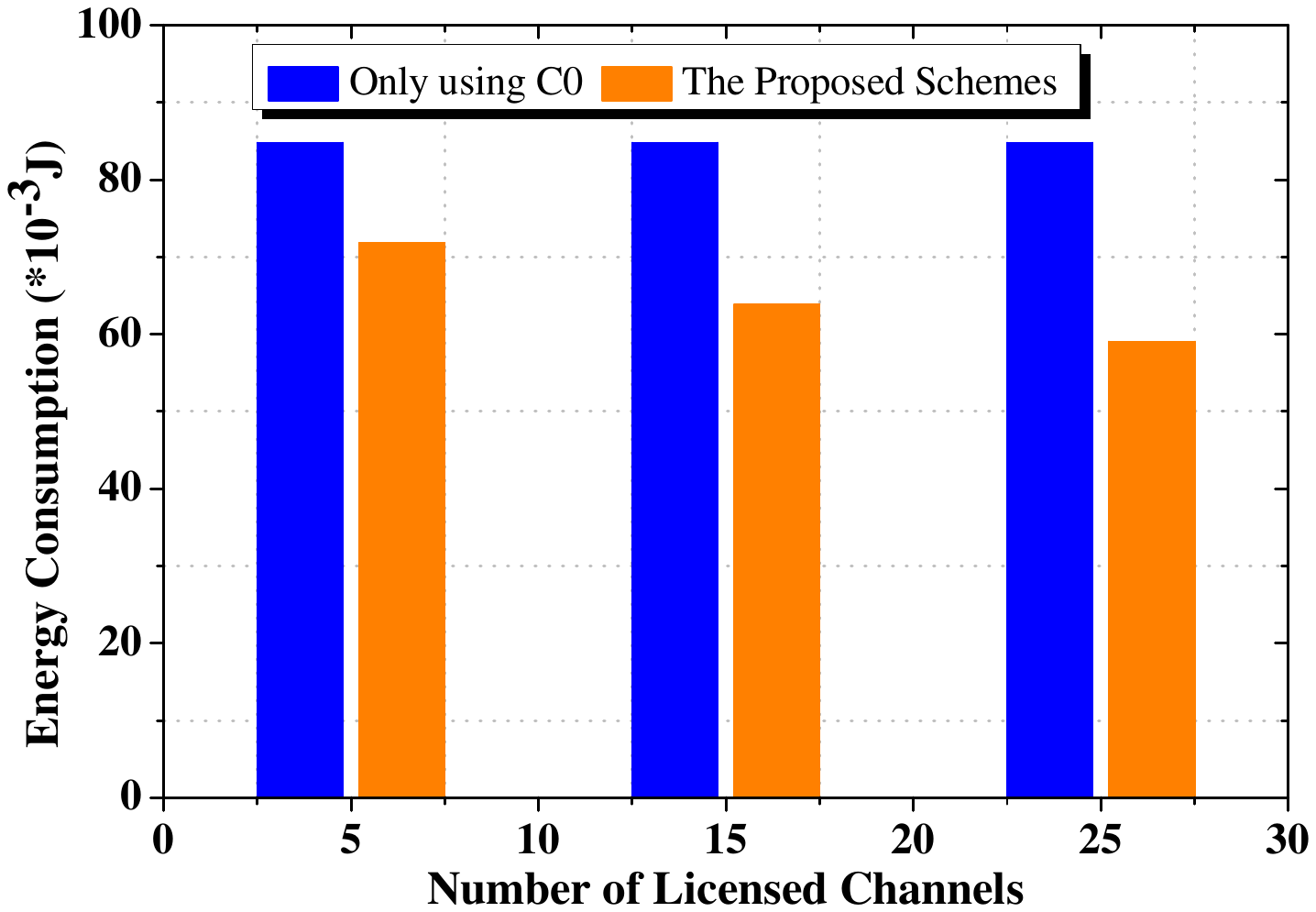}
\caption{Total energy consumption comparison under different numbers of licensed channels.}
\label{fig.7}
\end{figure}

\section{Conclusion}
\label{sec7}
In this paper, we have studied the opportunistic channel accessing problem to reduce the energy consumption in clustered CRSN, which is performed in a two-phase fashion. For intra-cluster data transmission, we have determined the condition to sense a licensed channel and proposed a dynamic channel sensing and accessing scheme for potential energy consumption reduction. For inter-cluster data transmission, we have also determined the channel sensing condition and jointly optimized the transmission power allocation and channel access to minimize the energy consumption. Extensive simulation results demonstrate that the proposed schemes can significantly reduce the energy consumption of data transmission and outperform the dynamic spectrum access schemes without considering the energy consumption of channel sensing and switching. For our future work, we will investigate rechargeable CRSNs, where the harvested energy will be carefully considered to support the cognitive radio techniques due to the stochastic energy arrival.

\ifCLASSOPTIONcaptionsoff
  \newpage
\fi

\bibliographystyle{IEEEtran}
\bibliography{Reference}

\end{document}